\documentclass[10pt,journal,a4paper]{IEEEtran}
\usepackage[utf8]{inputenc}
\usepackage{amsmath,amsfonts,amssymb,amsthm,bm,bbm}
\usepackage{graphicx,subfigure}
\usepackage{overpic}
\usepackage{enumitem}
\usepackage{cite}

\usepackage{hyperref}
\hypersetup{
  colorlinks   = true,    
  urlcolor     = blue,    
  linkcolor    = black,    
  citecolor    = red      
}

\author{Lorenzo Miretti~\IEEEmembership{Member,~IEEE}, Giuseppe Caire~\IEEEmembership{Fellow,~IEEE}, Sławomir Sta\'nczak~\IEEEmembership{Senior Member,~IEEE}
\thanks{L.~Miretti, G. Caire, and S.~Sta\'nczak are with the Technische Universität Berlin, Berlin 10587, Germany (email: \{miretti, caire, stanczak\}@tu-berlin.de). L.~Miretti and S.~Sta\'nczak are also with the Fraunhofer Institute for Telecommunications Heinrich-Hertz-Institut HHI, Berlin 10587, Germany. The authors acknowledge the financial support by the Federal Ministry of Education and Research of Germany in the programme of “Souverän. Digital. Vernetzt.” Joint project 6G-RIC, project identification numbers: 16KISK020K, 16KISK030.}
}

\title{Robust mmWave/sub-THz multi-connectivity using minimal coordination and coarse synchronization}

\newcommand{\eqdef}{:=}

\newcommand{\E}{\mathsf{E}}		
\renewcommand{\P}{\mathsf{Pr}} 			
\newcommand{\stdset}[1]{\mathbbmss{#1}}	
\newcommand{\set}[1]{\mathcal{#1}}		
\renewcommand{\vec}[1]{\bm{#1}}		
\newcommand{\CN}{\mathcal{CN}}			
\newcommand{\herm}{\mathsf{H}}			
\newcommand{\T}{\mathsf{T}}				

\newtheorem{lemma}{Lemma}

\newtheorem{proposition}{Proposition}

\newtheorem{remark}{Remark}
\newtheorem{corollary}{Corollary}

\allowdisplaybreaks

\begin{document}
\maketitle
\IEEEpubid{\begin{minipage}{\textwidth}\ \\[12pt] \centering
“© 2024 IEEE. Personal use of this material is permitted. Permission 
from IEEE must be obtained for all other uses, in any current or future 
media, including reprinting/republishing this material for advertising or 
promotional purposes, creating new collective works, for resale or 
redistribution to servers or lists, or reuse of any copyrighted 
component of this work in other works.
 \end{minipage}}

\begin{abstract}
This study investigates simpler alternatives to coherent joint transmission for supporting robust connectivity against signal blockage in mmWave/sub-THz access networks. By taking an information-theoretic viewpoint, we demonstrate analytically that with a careful design, full macrodiversity gains and significant SNR gains can be achieved through canonical receivers and minimal coordination and synchronization requirements at the infrastructure side. Our proposed scheme extends non-coherent joint transmission by employing a special form of diversity to counteract artificially induced deep fades that would otherwise make this technique often compare unfavorably against standard transmitter selection schemes. Additionally, the inclusion of an Alamouti-like space-time coding layer is shown to recover a significant fraction of the optimal performance. Our conclusions are based on a statistical single-user multi-point intermittent block fading channel model that, although simplified, enables rigorous ergodic and outage rate analysis, while also considering timing offsets due to imperfect delay compensation. In addition, we validate our theoretical approach by means of deterministic ray-tracing simulations that capture the essential features of next generation mmWave/sub-THz communications.
\end{abstract}

\section{Introduction}
A well-known major problem in mobile access networks operating at very high carrier frequencies, such as in the upper mmWave or sub-THz bands, is that they are highly sensitive to signal blockage \cite{keusgen2016shadowing,rappaport2017fading,sundeep2018blockage, aykin2019multibeam}. The reason is due to unavoidable physical characteristics of the propagation medium as well as the expected mode of operation, which relies on highly directional line-of-sight transmission in the noise limited regime \cite{rappaport2015wideband,song2020fully,miretti2023little}. This sensitivity causes connection instability and severely degrades the overall system performance. 
\IEEEpubidadjcol

In order to mitigate the signal blockage and provide robust connectivity, several multi-connectivity concepts are advocated by academia, industry, and standardization bodies \cite{rappaport2019diversity,3gpp2020multiconnectivity,qualcomm2021comp,ericsson2023concept}. In fact, a similar trend is also followed in the context of visible-light communication \cite{jungnickel2019optical}, where signal blockage is an evident issue. Essentially, all these concepts attempt to capitalize on the so-called \textit{macrodiversity} gains offered by simultaneous and possibly coordinated connections to multiple access points. The proposed technologies range from relatively simple control plane approaches for fast and efficient access point selection \cite{skouroumounis2017base,giordani2018control,gerasimenko2019capacity}, to data plane approaches that involve, for instance, concurrent transmission over orthogonal resources (i.e., parallel channels) \cite{fettweis2019reliable}, or even advanced \textit{network} multiple-input multiple-output (MIMO) techniques based on joint processing of distributed antenna systems \cite{tuninetti2016coverage,fettweis2020network,tolli2021blockage}. The present study considers variations of the latter technology, focusing on downlink transmission.
\begin{figure}[t!]
\centering
\includegraphics[width=0.7\linewidth]{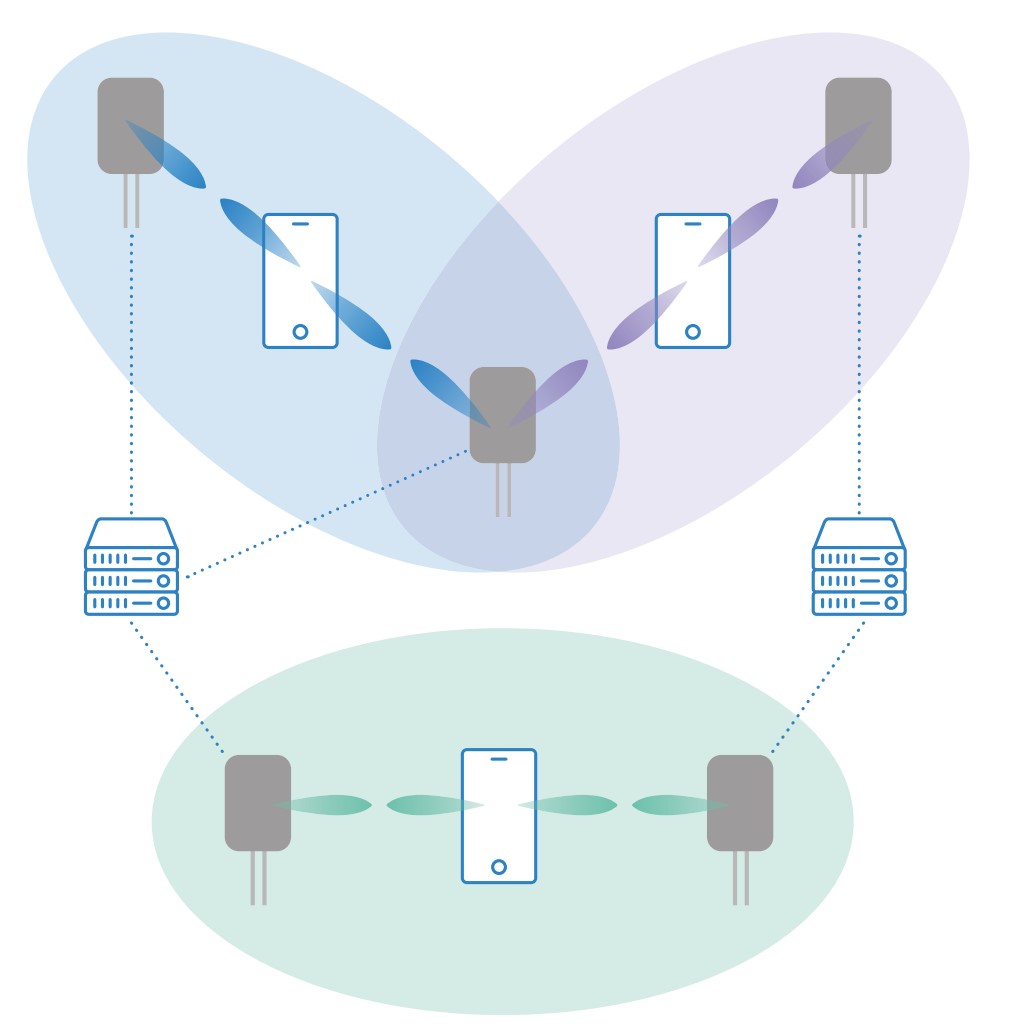}
\caption{Network of multiple transmitters simultaneously connected to multiple receivers via high-rate and directional mmWave/sub-THz links. }
\label{fig:model}
\end{figure}

\subsection{Summary of prior studies}
\IEEEpubidadjcol
The use of network MIMO techniques in distributed antenna systems has gained significant attention due to its potential to provide the best theoretical performance among all multi-connectivity concepts \cite{gesbert2010multicell}. However, the classical literature on network MIMO is typically based on ideal assumptions and often does not cover the specific characteristics of mmWave/sub-THz network MIMO systems, in particular in terms of channel models, hardware capabilities, and energy efficiency aspects. Therefore, it may lead to misleading conclusions for practical implementations. These limitations are addressed by several recent studies, focusing on different aspects of mmWave/sub-THz network MIMO systems. For example, \cite{tuninetti2016coverage} studies the benefits of network MIMO in terms of coverage using stochastic geometry tools, by assuming analog beamforming and different channel models. Furthermore, \cite{fettweis2020network} studies the benefits in terms of ergodic sum rate by assuming beam-domain statistical precoding, aided by a per-beam synchronization technique that significantly simplifies the system design. Moreover, robust precoding design against unknown signal blockage events is studied in \cite{tolli2021blockage}. 

More recently, a considerable number of studies started considering the extension to mmWave/sub-THz systems of the popular \textit{cell-free} massive MIMO framework \cite{demir2021foundations}, originally developed for sub-6GHz systems as a modern and refined version of network MIMO with special attention to time-division duplex operation, user-centric resource allocation, and scalability in terms of signal processing and system architecture. In this context, common investigated aspects are: energy efficiency \cite{alonzo2019energy, garcia2020energy, nguyen2022hybrid,guo2023energy,ismath2021deep,dasilva2023energy}; initial access and users to access points association  \cite{dandrea2021cell,zaher2022soft,pollin2021association,eryani2022self,ismath2021deep,wang2022joint, kamiwatari2022rf}; beam alignment and hybrid beamforming design 
\cite{femenias2019fronthaul,buzzi2022beam,bennis2022cellfree,femenias2022wideband,nguyen2022hybrid,feng2022weighted,zhang2020delay,pollin2021association,jafri2023cooperative,eryani2022self,ismath2021deep,wang2022joint,studer2022beam,kassam2022distributed,kim2023beam,kamiwatari2022rf}; 
limited capacity fronthaul \cite{femenias2019fronthaul,bennis2022cellfree,wang2022joint,kassam2022distributed}; impairments such as channel aging \cite{dandrea2021cell}, beam squint \cite{femenias2022wideband}, or low resolution data converters \cite{bennis2022cellfree}; and ray tracing simulations \cite{zaher2022soft,dasilva2023energy,ismath2021deep,studer2022beam}. A recurrent feature is that, motivated by the need for devices with low hardware complexity and energy consumption, the classical fully digital beamforming architecture is replaced by analog or hybrid beamforming architectures \cite{tuninetti2016coverage,fettweis2020network,alonzo2019energy,femenias2019fronthaul,dandrea2021cell,buzzi2022beam,bennis2022cellfree,femenias2022wideband,nguyen2022hybrid,feng2022weighted,zhang2020delay,pollin2021association,jafri2023cooperative,guo2023energy,eryani2022self,ismath2021deep,wang2022joint,studer2022beam,kassam2022distributed,kim2023beam,kamiwatari2022rf,zhu2024cellfree}, in line with most of the related literature on mmWave/sub-THz cellular systems. 

\subsection{Limitations of prior studies}
\label{ssec:prior}
A significant limitation of the available literature on network MIMO (and possibly cell-free) techniques for mmWave/sub-THz access is that it predominantly focuses on coherent joint transmission. In addition to large macrodiversity and SNR gains, this technique offers particularly advanced interference management capabilities. However, it requires tight synchronization and accurate channel estimation at the transmitters. Unfortunately, achieving sufficient synchronization and channel estimation accuracy for coherent joint transmission is notoriously a highly non-trivial task \cite{tuninetti2016coverage,fettweis2020network}, even for sub-6GHz systems \cite{irmer2011coordinated}. Due to the significantly higher carrier frequency, bandwidth, and lower quality hardware, this issue can only worsen in mmWave/sub-THz systems. Hence, it is expected that phase coherent joint transmission in mmWave/sub-THz systems will likely be unfeasible, especially in the early generations of such systems. Note that  the per-beam synchronization technique proposed in \cite{fettweis2020network} is used for selecting more efficient orthogonal frequency-division multiplexing (OFDM) parameters than in traditional approaches, but not to resolve the fundamental issues preventing coherent joint transmission. Nevertheless, since mmWave/sub-THz systems are envisioned to operate in the noise-limited regime with ultra-large bandwidths, advanced interference management is of secondary importance. In this context, alternative techniques that do not require tightly synchronized access points and accurate global channel estimation, but still deliver significant macrodiversity and SNR gains, become particularly appealing.

To the best of our knowledge, alternative techniques to coherent joint transmission are reported only in \cite{tuninetti2016coverage,dandrea2021cell,buzzi2022beam,fettweis2020network}. Specifically, \cite{tuninetti2016coverage,dandrea2021cell,buzzi2022beam} consider non-coherent joint transmission of a single statistically beamformed data stream per user, i.e., of standard scalar codewords. However, this technique effectively induces strong multipath components that may sum up non-coherently at the receiever. As we will see, the impact of this artificially induced fading may be quite severe in several practical regimes of interest, at the point that non-coherent joint transmission may compare unfavorably with respect to simpler alternatives based on transmitter selection \cite{skouroumounis2017base,giordani2018control,gerasimenko2019capacity}. In contrast, \cite{fettweis2020network} considers spatially multiplexed and statistically beamformed joint transmission of a possibly very large number (theoretically up to the total number of transmit antennas) of data streams per user, i.e., vector-valued codewords. As shown in \cite{fettweis2020network} and in this study, this technique may have remarkable theoretical performance. However, \cite{fettweis2020network} does not discuss the practical implications of choosing this technique. For example, \cite{fettweis2020network} does not discuss how to achieve the theoretical gains using decoding algorithms with practical complexity, especially in the common regime where the number of antennas (or radio frequency chains) at the receiver is smaller than the number of transmitted data streams, i.e., where linear processing cannot demultiplex the data streams. Furthermore, \cite{fettweis2020network} does not compare the performance of the proposed technique against other simpler alternatives such as the ones based on transmitter selection. These aspects are crucial because, given the extremely high rates of mmWave/sub-THz communication systems, low-complexity processing is of primary importance.

\subsection{Contributions of this study}
In this study, we give the first comprehensive investigation of alternative techniques to coherent joint transmission tailored to mmWave/sub-THz network MIMO systems. In particular, we review and extend several available techniques and compare them in terms of macrodiversity gains, SNR gains, as well as the required network infrastructure and receiver architecture. In contrast to prior studies, we adopt an information theoretic viewpoint and, starting from first principles, we study the problem of reliable communication over a simple single-user multi-point \textit{intermittent} block fading channel. Despite being considerably simplified, we demonstrate via realistic ray-tracing simulations that the proposed model still captures the most essential features of the envisioned mmWave/sub-THz systems, such as line-of-sight transmission in the noise limited regime and sensitivity to signal blockages. 

Our main result is the identification of a family of transmission techniques that extends non-coherent joint transmission as a promising candidate for supporting robust transmission with a simpler network infrastructure and receiver architecture than the ones required by the competing techniques outlined in Section~\ref{ssec:prior}. More precisely, we demonstrate analytically that with a careful design, full macrodiversity gains and significant SNR gains can be achieved through the use of linear receiver processing, standard scalar point-to-point coding, minimal coordination across the transmitters, and coarse synchronization. Here, by minimal coordination, we mean that the signals are jointly formed either without any instantaneous channel knowledge, or based on knowledge of blockage events at most. To this end, a first contribution is the study of a special form of diversity, called \textit{phase} diversity \cite{dammann2002low}, which is particularly suitable for the considered multi-point intermittent fading model and makes the effective channel fluctuations essentially driven by blockage events only. A second contribution is the study of a simple \textit{space-time} coding layer that extends the well-known Alamouti scheme \cite{alamouti1998} to provide an additional performance boost. Finally, a third contribution is the study of the impact of unknown timing offsets based on a worst-case approach similar to the information theoretical literature on asynchronous \cite{cover1982asynchronous} or arbitrarily varying \cite{blackwell1960capacities} channels.

\subsection{Outline and notation}
The rest of this study is organized as follows. Section~\ref{sec:model} presents the proposed channel model and its fundamental limits. Section~\ref{sec:schemes} studies and extends several suboptimal transmission schemes for the considered channel, and provides practical design guidelines. Section~\ref{sec:time} shows how to apply the obtained insights to imperfect time synchronization and delay compensation. In Section~\ref{sec:sim}, our approach is validated by means of ray-tracing simulations. Section~\ref{sec:conclusion} summarizes the results and discusses promising research directions.

We use $\stdset{R}$, $\stdset{C}$, and $\stdset{Z}$ to denote, respectively, the set of real numbers, complex numbers, and integers. Boldface lower case letters, $\vec{a}$, denote column vectors, and boldface upper case letters, $\vec{A}$, denote matrices. The $n$th entry of $\vec{a}$ is denoted by $a_n$. We use $\eqdef$ for definitions. We denote by $\vec{a}^\T$ and $\vec{a}^\herm$ the transpose and Hermitian transpose of $\vec{a}$, respectively. We use $\P(\cdot)$ for the probability of an event, and $\E[\cdot]$ for the expectation operator. Given a stationary random process $\{\vec{a}_t\}_{t\in \stdset{Z}}$, we denote an arbitrary realization by $\vec{a}$. All (in)equalities involving random variables should be interpreted as almost sure (in)equalities.

\section{Intermittent block fading channels}
\label{sec:model}
In this section, we present and study the fundamental limits of the proposed multi-point intermittent block fading channel model that forms the basis of our analysis. 
 
\subsection{System architecture}
We consider the downlink of a mmWave/sub-THz network where each receiver is jointly served by multiple transmitters in the same time-frequency resource, and possibly in a user-centric fashion (see Figure~\ref{fig:model}). We assume an analog beamforming architecture in which each transmitter is equipped with a single digital antenna port driving a highly directional steerable beamforming antenna. We further assume the receiver to be equipped with a single digital antenna port driving either an  omnidirectional antenna, or a steerable multi-lobe analog beamforming antenna \cite{aykin2019multibeam}. As customary, line-of-sight directional transmission is employed to counteract the large path-loss at mmWave/sub-THz frequencies. Moreover, we focus on a single-user transmission model, i.e., we do not consider multi-user interference, as done in, e.g., \cite{aykin2019multibeam,zhu2024cellfree}, and in nearly all studies assuming standard orthogonal multiple access based on time and/or frequency division multiplexing.  

Although more advanced receivers may be equipped with multiple digital antenna ports, the single receiver port case is an insightful example of the likely scenario where the number of receive ports is much smaller than the total number of transmit ports. In addition, our model could also be extended to a modular hybrid beamforming architecture where each transmitter is equipped with a small number of digital antenna ports, each serving a different receiver through a separate highly directional steerable beamforming antenna.  In fact, if resources (time, frequency, and highly directional beamforming antennas) are properly allocated, the residual interference can be typically neglected in the noise-limited regime. The main advantage of this architecture is that it offers relatively good spatial multiplexing capabilities while keeping hardware complexity and synchronization requirements low. However, the underlying resource allocation problems are beyond the scope of this study. Hence, the  analysis of the above modular architecture is left as a promising future research direction. \color{black} 

\subsection{Channel model}
To facilitate analytical treatability, while still capturing the essence of the problem, we model the time-domain downlink input-output relation between an arbitrary receiver and $L$ transmitters with comparable path loss using the following simple multi-point intermittent block fading model $(\forall m \in \stdset{Z})$:
\begin{equation*}
y[m] = \sum_{l=1}^Lh_{l,t}x_l[m] + z[m],\quad h_{l,t} = \beta_{l,t} e^{j\theta_{l,t}}, \quad t=\left\lfloor \frac{m}{T}\right\rfloor,
\end{equation*}
where $y[m]\in \stdset{C}$ is the received signal, $x_l[m]\in \stdset{C}$ is the signal of the $l$th transmitter, $z[m]\sim \CN(0,1)$ is a sample of a white Gaussian noise process, $\{\theta_{l,t}\}_{t\in\stdset{Z}}$ is an independent stationary and ergodic random process with first order distribution $\text{Uniform}(0,2\pi)$ modelling random phase offsets for the $l$th transmitter, and $\{\beta_{l,t}\}_{t\in\stdset{Z}}$ is an independent (stationary and ergodic) binary Markov chain modelling random blockage events for the $l$th transmitter, parametrized by the transition probabilities from connected to blocked state $p := \P(\beta_{l,t+i} = 0|\beta_{l,t} = 1)<1$, and from blocked to connected state $q := \P(\beta_{l,t+i} = 1|\beta_{l,t} = 0)<1$. The state $\vec{h}_t:=[h_{1,t},\ldots,h_{L,t}]^\T$ stays constant for a block of $T\gg 1$ channel uses, and evolves from block to block according to a stationary and ergodic random process indexed by $t\in \stdset{Z}$.

Although simplified, the proposed model captures the essential features of the envisioned mmWave/sub-THz network while allowing us to gain analytical insights for the design of practical systems. The main underlying assumptions are listed and motivated below:
\begin{itemize}[leftmargin=*]
\item We neglect multipath propagation and consider only the line-of-sight paths between each transmitter and the receiver, as, e.g., in \cite{tuninetti2016coverage}. This is supported by several measurement campaigns showing  that, in typical environments, reflected multipath components are often significantly weaker than the line-of-sight component when using directional line-of-sight transmission \cite{molish2022directionally,alper2023uniform,miretti2023little}. Therefore, these components can be neglected in the noise limited regime \cite{miretti2023little}.
\item We model blockages statistically, using a somewhat artificial binary Markov chain. This is different than the deterministic worst-case approach followed, e.g., in \cite{tolli2021blockage}. The channel modeling community has suggested that, although simplified, similar statistical models based on Markov chains can be instrumental for studying the performance of mmWave/sub-THz systems subject to rapid blockages (see, e.g., \cite{rappaport2017fading,rappaport2019NYUSIM}). Following a similar approach, our binary state model is motivated by the observation (also confirmed by our realistic simulations in Section~\ref{sec:sim}) that practical path loss variations induced by blockage events are often so large that the useful signal power becomes negligible in the noise-limited regime, effectively leading to an 'on-off' behavior. The Markov property is introduced as it leads to a stationary and ergodic random state process, enabling statistical analysis in terms of (first-order) macroscopic parameters such as the blockage probability, which are assumed constant over a sufficiently long period compared to the duration of transmission. Note that stationarity and ergodicity are required in nearly every practically useful information-theoretic study on channels with memory. For example, we point out that a binary Markov chain is commonly used by the information-theoretic community for developing coding techniques for reliable communication over channels that experience sudden transitions between a `good' and a `bad' state, such as in the well-known Gilbert-Elliot model \cite{mushkin1989gilbert}.  Moreover, we remark that the assumption of independent blockages across multiple transmitters is reasonable when blockages are caused by events close to the transmitters \cite{rappaport2019diversity}. Multi-state models and/or correlated blockages may also be investigated, but such analysis is left for future work.
\item We use random phase shifts unkown at the transmitters to model both small-scale channel variations and the lack of phase synchronization at the transmitters, as, e.g., in \cite{tuninetti2016coverage}. This is motivated by the known challenges in compensating channel phases and local oscillators drifts at the transmitters.
\item We assume perfect time synchronization and delay compensation. While this may not be a realistic assumption, later in Section~\ref{sec:time} we show that the insights obtained with our ideal model can be applied also to more realistic systems where the transmitters can only achieve coarse time synchronization and delay compensation. 
\end{itemize}

We conclude this section by recalling some useful definitions and properties related to the joint blockage process $\{\vec{\beta}_t\}_{t\in \stdset{Z}}:=\{(\beta_{1,t},\ldots,\beta_{L,t})\}_{t\in \stdset{Z}}$. Following known properties of the stationary distribution of binary Markov chains, we define the blockage probability $(\forall l \in \{1,\ldots,L\})$   
\begin{equation*}
p_B \eqdef \P(\beta_l = 0) = \frac{p}{p+q},
\end{equation*}
where we recall that $\beta_l$ denotes a realization of $\beta_{l,t}$. Since the joint blockage process is governed by $L$ i.i.d. binary Markov chains, its stationary distribution is readily given in terms of the above quantities as $
\P[\vec{\beta} = (b_1,\ldots,b_L)] = \prod_{l=1}^L \P(\beta_l=b_l)$. We also define the number of non-blocked transmitters $\alpha_t := \sum_{l=1}^L\beta_{l,t}$, and observe that $\alpha \sim \text{Binomial}(L,1-p_B)$. We readily have that $
\P(\alpha \geq 1) = \P(\vec{\beta} \neq \vec{0}) = 1-p_B^L$, which formalizes one key benefit of macrodiversity, that is, the probability of a receiver to experience complete blockage decreases with $L$. 

\subsection{Capacity and optimal transmission}
In this section, we study the ergodic and outage capacity of the considered channel, defined as follows using the information-theoretic notion of (not necessarily memoryless) state-dependent channel with causal state information at the transmitter (see, e.g., \cite{biglieri1998fading,caire1999capacity}). We first assume that the source messages are available at all transmitters, as in most studied based on transmitter cooperation. Moreover, we define the channel state $(\vec{\beta},\vec{\theta})$ and assume perfect channel state information at the receiver (CSIR), as in most studies based on coherent demodulation and decoding. Furthermore, we assume that all transmitters know $\vec{\beta}$, i.e., we assume centralized partial channel state information at the transmitters (CSIT). The term \textit{centralized} is due to the fact that this is equivalent to a setup where the signals of the different transmitters are formed by a central controller endowed with CSIT $\vec{\beta}$. We do not consider the case of different channel state information at each transmitter, a configuration known as \textit{distributed} CSIT \cite{miretti2021cooperative}. The term \textit{partial} refers to the fact that the transmitters do not know the full state $(\vec{\beta},\vec{\theta})$. Note that, in contrast to the phases $\vec{\theta}$, the blockages $\vec{\beta}$ are typically slowly-varying macroscopic parameters ($p,q \ll 1$). Hence, they can be easily estimated at the transmitter side, e.g., by means of uplink pilot signals in both time-division and frequency-division duplex systems. Finally, we assume an instantaneous power constraint $P\in \stdset{R}_{+}$ per transmitter. In absence of latency constraints, standard arguments show that a well-defined ergodic capacity in the classical Shannon sense exists \cite{caire1999capacity}, and it can be expressed (in bits/symbol) as 
\begin{equation*}
C := \sup_{\vec{Q}\in \set{Q}}\E[\log(1+\vec{h}^\herm \vec{Q}(\vec{\beta}) \vec{h})],
\end{equation*}
where $\set{Q}$ denotes the set of functions mapping each CSIT realization $\vec{\beta}$ to a complex-valued symmetric positive semidefinite matrix $\vec{Q}(\vec{\beta})$ of size $L\times L$, with diagonal entries satisfying 
\begin{equation*}
(\forall l \in \{1,\ldots,L\})~[\vec{Q}(\vec{\beta})]_{l,l} \leq P.
\end{equation*} 
For a given $\vec{Q}\in \set{Q}$, the rate $R = \E[\log(1+\vec{h}^\herm \vec{Q}(\vec{\beta}) \vec{h})]$ is achievable by vector Gaussian codes with conditional input covariance $\vec{Q}(\vec{\beta})$, i.e., by letting $(\forall m \in \stdset{Z})$
\begin{equation}\label{eq:tx_signal}
\begin{bmatrix}
x_1[m] &
\ldots &
x_L[m] 
\end{bmatrix}^\T = \vec{Q}(\vec{\beta}_t)^{\frac{1}{2}}\vec{u}[m], \quad t=\left\lfloor \frac{m}{T}\right\rfloor.
\end{equation}
where $\vec{u}[m] \sim \CN(\vec{0},\vec{I})$ is a vector-valued i.i.d. information bearing signal. 

\begin{proposition}\label{prop:C}
The ergodic capacity of the considered channel assuming partial CSIT $\vec{\beta}$ and an instantaneous power constraint $P\in \stdset{R}_+$ per transmitter is given by
\begin{equation}\label{eq:C}
C = \E[\log(1+\alpha P)],
\end{equation}
where $\alpha = \sum_{l=1}^L\beta_l \sim \text{Binomial}(L,1-p_B)$.
\end{proposition}
\begin{proof}
By the law of total expectation and Jensen's inequality, we obtain the upper bound $(\forall \vec{Q}\in \set{Q})$
\begin{align*}
\E[\log(1+\vec{h}^\herm \vec{Q}(\vec{\beta}) \vec{h})] 
 & \leq \E[\log(1+\E[\vec{h}^\herm \vec{Q}(\vec{\beta}) \vec{h}|\vec{\beta}])] \\
 & = \E[\log(1+\mathrm{tr}(\E[\vec{h} \vec{h}^\herm  |\vec{\beta}]\vec{Q}(\vec{\beta})))] \\
 & = \E[\log(1+\mathrm{tr}(\mathrm{diag}(\vec{\beta})\vec{Q}(\vec{\beta})))] \\
 & = \E\left[\log\left(1+\sum_{l=1}^L\beta_l[\vec{Q}(\vec{\beta})]_{l,l}\right)\right] \\
 & \leq \E[\log(1+\alpha P)].
\end{align*}
Choosing $\vec{Q}(\vec{\beta})=P\vec{I}$ concludes the proof. 
\end{proof}
\begin{corollary}
The proof of Proposition~\ref{prop:C} also shows that the ergodic capacity of the considered channel assuming no CSIT and an instantaneous power constraint $P\in \stdset{R}_+$ per transmitter is given by
\begin{equation*}
C = \E[\log(1+\alpha P)].
\end{equation*}
i.e., knowledge of the blockage process $\vec{\beta}$ at the transmitters does not increase capacity.
\end{corollary}
The above analysis shows that the ergodic capacity \eqref{eq:C} can be achieved by the transmission scheme in \eqref{eq:tx_signal} with fixed diagonal input covariance independent of $\vec{\beta}$. Note that this is essentially the same transmission scheme considered in \cite{fettweis2020network}, left apart a power allocation step. However, although optimal from an information-theoretic viewpoint, this scheme presents several practical drawbacks. First and foremost, being equipped with a single antenna, the receiver cannot demultiplex the vector-valued information bearing signal using linear processing, and hence it must perform complex algorithms such as maximum-likelihood vector decoding or approximate versions based, e.g., on successive interference cancellation. Second, although \eqref{eq:C} can be theoretically achieved without CSIT, this would require codewords spanning a very large number of blockage realizations, hence largely exceeding practical latency constraints. Therefore, to reduce latency, practical systems should typically employ block codes of maximum length $T$. By interpreting $\log(1+\alpha P)$ as the maximum achievable rate in a given block-fading realization, the ergodic capacity \eqref{eq:C} can be then approached by means of rate adaptation mechanisms based on CSIT $\alpha$. Note that an intermediate solution in terms of latency is to use hybrid automatic retransmission-request (HARQ) mechanisms \cite{sesia2004incremental}, which require some feedback to the transmitters. 

\begin{figure}[ht!]
\centering
\includegraphics[width=1\linewidth]{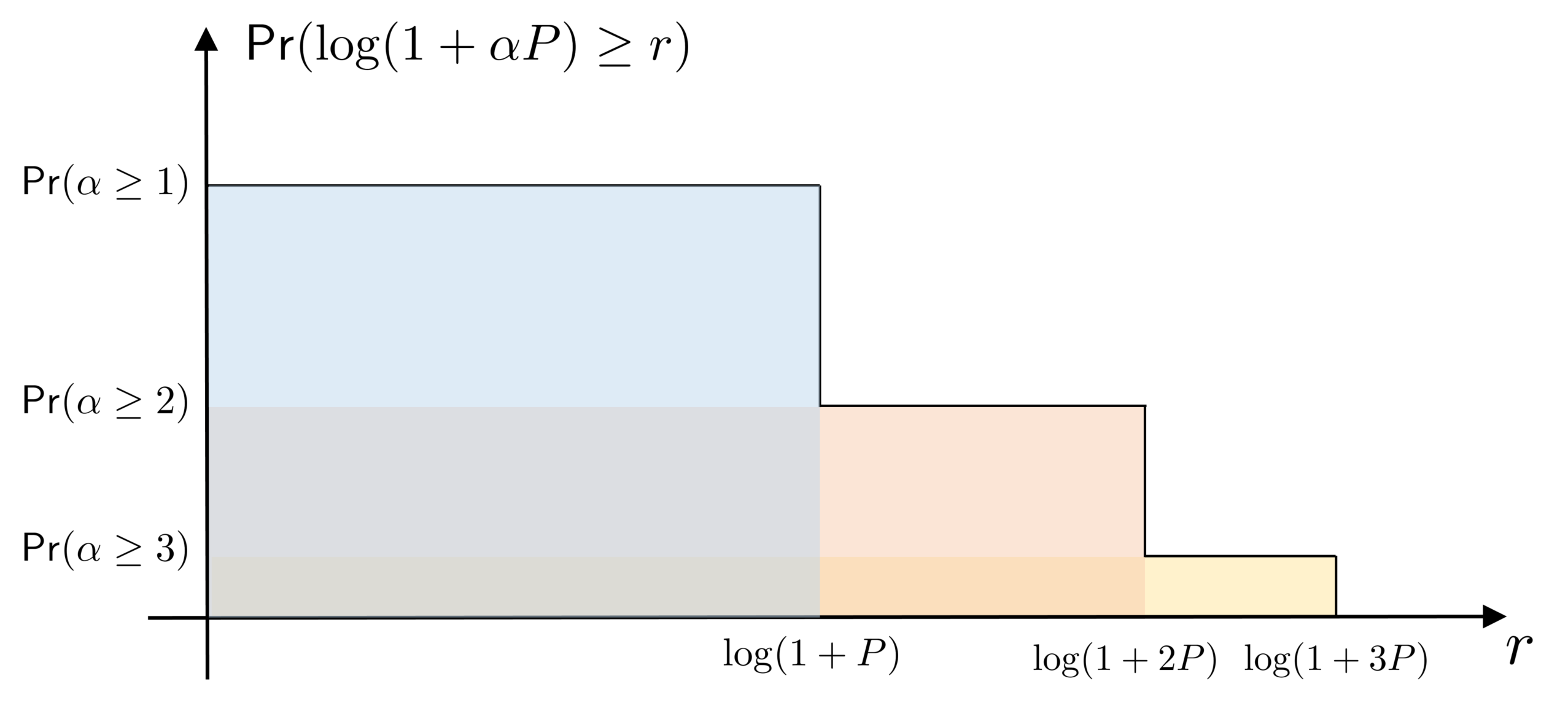}
\caption{Pictorial representation of the CCDF of the instantaneous capacity.}
\label{fig:ccdf}
\end{figure}
%

As a more appropriate performance metric in the presence of strict latency constraints and when rate adaption mechanisms based on CSIT are prohibitive, we now study the outage capacity of the considered channel. In particular, by considering fixed-rate block codes of maximum length $T$, we measure capacity as
\begin{align*}
C_{\mathrm{out}} := \sup_{r\in \stdset{R}} r\cdot \P\left(\log(1+\alpha P) \geq r\right).
\end{align*}
\begin{proposition}\label{prop:C_out}
The outage capacity of the considered channel with power constraint $P\in \stdset{R}_+$ per transmitter is given by
\begin{equation}\label{eq:C_out}
C_{\mathrm{out}} = \max_{i\in\{1,\ldots,L\}} \P(\alpha \geq i) \log(1+iP).
\end{equation}
\end{proposition}
\begin{proof}
The proof readily follows by observing that $\P\left(\log(1+\alpha P) \geq r\right)$ is a decreasing step function of $r$ with discontinuities located at $\log(1+iP)$ and step heights $\P(\alpha \geq i)$, for $i\in \{1,\ldots,L\}$.
\end{proof}
The quantitative difference between the ergodic and outage capacity can be visualized in Figure~\ref{fig:ccdf}, which depicts the complementary cumulative density function (CCDF) of the instantaneous capacity $\log(1+\alpha P)$. The ergodic capacity is given by the total area underlying the CCDF, while the outage capacity is given by the area $\P(\alpha \geq 1) \log(1+P)$ of the largest rectangle. 

%
%

\section{Practical transmission schemes}
\label{sec:schemes}
In this section we revisit, compare, and extend a set of practical transmission schemes that do not require the receiver to perform complex decoding algorithms. We will measure performance in terms of ergodic and outage rates, depending on the context and applicability, and discuss crucial implementation aspects. Our theoretical results and the related supporting illustrations in this section are all based on the same simplified channel model in Section~\ref{sec:model}, and they are completely characterized by the parameters $p_B$ (blockage probability), $P$ (per-transmitter SNR), and $L$ (number of transmitters).

\subsection{Transmitter selection}
As first main baseline, we consider a transmission scheme where, for each $t$th block-fading realization with at least one non-blocked transmitters, i.e., such that $\alpha_t > 0$, the network chooses one non-blocked transmitter $l^\star(\vec{\beta}_t)$ and lets $(\forall m \in \stdset{Z})$ $(\forall l \in \{1,\ldots,L\})$
\begin{equation*}
x_l[m] = \begin{cases} u[m] & \text{if } l = l^\star(\vec{\beta}_t) \text{ and } \alpha_t>0, \\ 0 & \text{otherwise,}\end{cases} \quad t=\left\lfloor \frac{m}{T}\right\rfloor,
\end{equation*}
where $u[m]\sim \CN(0,P)$ is a scalar i.i.d. information bearing signal. 
By mapping this scheme to some $\vec{Q}\in \set{Q}$ in the signal model \eqref{eq:tx_signal}, it can be readily seen that the ergodic rate 
\begin{equation*}
R = (1-p_B^L)\log(1+P)
\end{equation*} 
is achievable. Note that the concept of outage rate is not meaningful for this scheme, since the outage events are known by the network. The main advantage of transmitter selection is that it can mitigate the effect of blockage using simple scalar codes at fixed rate $\log(1+P)$. However, this simplicity comes with two major drawbacks. First, no SNR gain is provided, i.e., $R$ saturates to $\log(1+P)$ as $L$ grows. Second, causal knowledge of $\vec{\beta}_t$ at the transmitters is assumed. Alternatively, the choice of the transmitter can be performed by the receiver, but this procedure still entails the timely feedback of $l^\star(\vec{\beta}_t)$ to the network. Overall, the required resources for timely blockage estimation and network coordination may be significant, especially in case of frequent blockages.  

\subsection{Non-coherent joint transmission}
As second main baseline, we consider non-coherent joint transmission of a single scalar codeword from all transmitters simultaneously, i.e., we let $(\forall m \in \stdset{Z})$ $(\forall l\in \{1,\ldots,L\})$
\begin{equation*}
x_l[m] = u[m],
\end{equation*}
where $u[m]\sim \CN(0,P)$ is a scalar i.i.d. information bearing signal.
Using \eqref{eq:tx_signal}, we observe that this scheme achieves the ergodic rate
\begin{equation}\label{eq:R_NCJT}
R = \E[\log(1+|h|^2P)],
\end{equation}
where $h \eqdef \sum_{l=1}^L\beta_le^{j\theta_l}$ denotes an effective small-scale fading coefficient which is artificially induced by the non-coherent superposition of signals at the receiver. 

Non-coherent joint transmission always outperforms transmitter selection in terms of ergodic rates, and it provides SNR gains as $L$ grows (note that we assume a per-transmitter power constraints, hence the total power budget grows unboundedly with $L$). This trend is clearly visible in Figure~\ref{fig:R_erg}. For a formal proof, we exploit the following intuitive property, which essentially states that, on average, adding non-blocked transmitters is beneficial.
\begin{proposition}\label{prop:monotonicity}
Let $\{\theta_l\}_{l=1}^\infty$ be an i.i.d. random process with first order distribution $\mathrm{Uniform}(0,2\pi)$. Then, the sequence $\{\bar{R}(i)\}_{i=0}^\infty$ given by $(\forall i \geq 0)$ 
\begin{equation}\label{eq:barR}
\bar{R}(i) \eqdef \E\left[\log\left(1+\left|\textstyle\sum_{l=1}^ie^{j\theta_l}\right|^2P\right)\right]
\end{equation}
is strictly increasing and unbounded above, i.e., 
\begin{equation*}
(\forall i\geq 0)~\bar{R}(i+1) > \bar{R}(i), \text{ and }\lim_{i\to \infty}\bar{R}(i)= \infty.
\end{equation*} 
\end{proposition}
\begin{proof}
The proof is given in Appendix~\ref{app:monotonicity}. 
\end{proof}
The gains of non-coherent joint transmission can be then formalized as follows.
\begin{proposition}\label{prop:R_NCJT}
The ergodic rate in \eqref{eq:R_NCJT} achieved by non-coherent joint transmission satisfies: 
\begin{enumerate}[label = (\roman*)]
\item $(\forall L\geq 1)~R \geq (1-p_B^L)\log(1+P)$;
\item $R$ is a strictly increasing sequence in $L$;
\item $\lim_{L\to \infty} R = \infty$.
\end{enumerate}
\end{proposition}
\begin{proof}
The proof is given in Appendix~\ref{app:R_NCJT}.
\end{proof}
However, these potential gains may be very difficult to achieve in practice. In fact, although no CSIT is theoretically needed, practical coding schemes based on rate adaptation or HARQ mechanisms need to track and adapt to the instantaneous rate fluctuations $\log(1+|h|^2P)$. Unfortunately, the instantaneous rate $\log(1+|h|^2P)$ fluctuates at a much faster pace and over a much larger dynamic range than for the case of transmitter selection, where the fluctuations are driven by the blockage process only. Therefore, these mechanisms may consume a significantly larger amount of resources (e.g., feedback) than in the case of transmitter selection. As a consequence, despite the higher ergodic rates, non-coherent joint transmission may not be preferable over transmitter selection from a overall system engineering perspective.

To study the performance of a more practical non-adaptive low-latency implementation of non-coherent joint transmission, we consider the outage rate
\begin{equation*}
R_{\mathrm{out}} = \sup_{r\in \stdset{R}} r\cdot \P\left(\log(1+|h|^2 P) \geq r\right).
\end{equation*}
The main advantage of this implementation is that it requires minimal coordination effort at the network side. However, our numerical results show that the outages due to the artificially induced small-scale fading can be significant, and that there are non-trivial dependencies on the system parameters. In particular, by comparing Figure~\ref{fig:R_out} and Figure~\ref{fig:R_erg}, we can see that $R_{\mathrm{out}}$ is well below $(1-p_B^L)\log(1+P)$ (i.e., the performance of transmitter selection) in most cases of interest. Hence, the simplicity of this non-coherent joint transmission implementation may come at a significant price in terms of performance, which is satisfactory for very specific cases only, such as the case of a very large number of transmitters. 

\subsection{Phase diversity}
To address the limitations of non-coherent joint transmission, in this section we propose an extension (based on so-called \textit{phase} diversity \cite{dammann2002low}) that mitigates the detrimental impact of the artificially induced small-scale fading. The main idea is to induce a controlled fast-fading regime for escaping deep fading bursts. We split each fading block of length $T$ into an integer number $M = T/K$ of frames composed by $K$ symbols each. For convenience, we rearrange the time-domain channel model onto a grid model (similar to an OFDM grid) as follows: $(\forall m \in \stdset{Z})$ $(\forall k\in \{1,\ldots,K\})$ 
\begin{equation}\label{eq:model_phase_div}
y[m,k] = \sum_{l=1}^Lh_{l,t}x_l[m,k] +z[m,k],\quad t=\left\lfloor \frac{m}{M}\right\rfloor. 
\end{equation}
We then let $(\forall m \in \stdset{Z})$ $(\forall k\in \{1,\ldots,K\})$ $(\forall l \in \{1,\ldots,L\})$
\begin{equation}\label{eq:signal_phase_div}
x_l[m,k] = e^{j\phi_{l,k}}u[m,k],
\end{equation}
where $u[m,k]\sim \CN(0,P)$ is a scalar i.i.d. information bearing signal, and where $(\forall l \in\{1,\ldots,L\})$ $(\phi_{l,1},\ldots,\phi_{l,K})$ is a vector of random phases independently and uniformly distributed in $[0,2\pi]$, which is shared among the transmitters and receiver as a common source of randomness. Provided that the number of frames $M$ is large enough, standard arguments (as the ones used, e.g., for coding in OFDM systems) show that the following outage rate is achievable:
\begin{equation*}
R_{\mathrm{out}} = \sup_{r\in \stdset{R}} r\cdot \P\left(\dfrac{1}{K}\sum_{k=1}^K\log(1+|h[k]|^2P) \geq r\right),
\end{equation*}
where $h[k] \eqdef \sum_{l=1}^L\beta_l e^{j(\theta_l+\phi_{l,k})}$ denotes an effective small-scale fading coefficient that fluctuates within the same frame of $K$ symbols, according to the given vector of random phases. The proposed extension is motivated by the following asymptotic property:

\begin{proposition}\label{prop:lln} For all $i\in \{0,\ldots,L\}$, let $\bar{R}(i)$ as in \eqref{eq:barR}. Then, $(\forall \epsilon > 0)$
\begin{align*}
\lim_{K\to \infty}\P\left(\left|\dfrac{1}{K}\sum_{k=1}^{K}\log(1+|h[k]|^2P)-\bar{R}(\alpha) \right|\geq \epsilon\right)=0,
\end{align*}
i.e., the instantaneous rate converges in probability to $\bar{R}(\alpha)$ as $K$ grows large.
\end{proposition}
\begin{proof}
The proof is given in Appendix~\ref{proof:lln}. 
\end{proof}

\begin{figure}[ht!]
\centering
\includegraphics[width=0.9\linewidth]{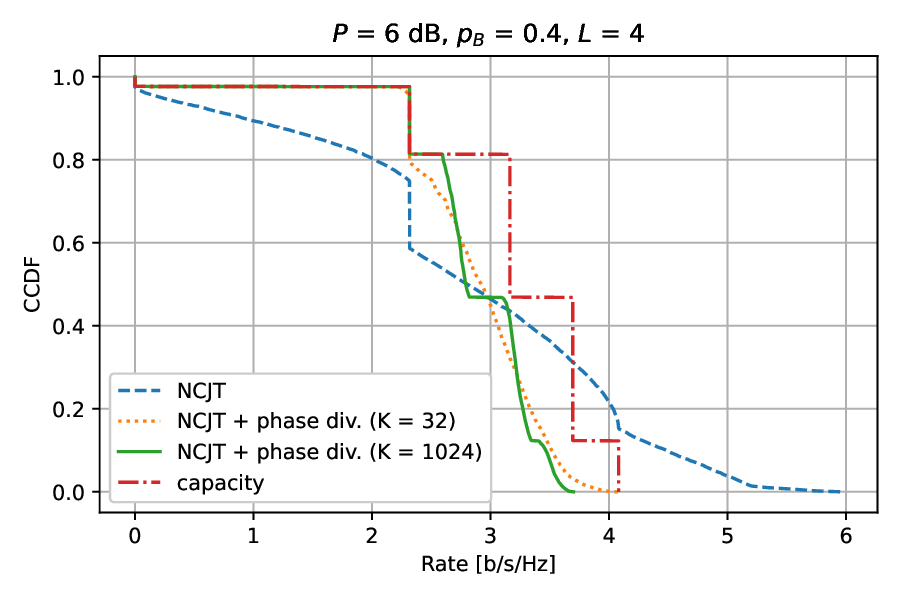}
\caption{CCDF of the instantaneous capacity, and of the instantaneous rate of non-coherent joint transmission (NCJT) with and without phase diversity, for a given choice of the parameters $(p_B,P,L)$ of the simplified channel model in Section~\ref{sec:model}.}
\label{fig:ccdf_NCJT_phasediv}
\end{figure}

The above proposition suggests that the proposed extension of non-coherent joint transmission can reduce the effect of the artificially induced small-scale fading, and make the instantaneous rate fluctuations be essentially driven by the aggregate blockage process $\alpha$. This effect is clearly visible in Figure~\ref{fig:ccdf_NCJT_phasediv}, where the CCDF of the instantaneous rate approaches $\P\left(\bar{R}(\alpha)\geq r\right)$, i.e., a step-wise function similar to the CCDF of the instantaneous capacity $\log(1+\alpha P)$, as $K$ grows. The  difference between these two functions is a Jensen's penalty for the discontinuities indexed by $i\geq 2$, since $(\forall i\geq 1)$ $\bar{R}(i) \leq  \log(1+ iP)$ holds, with equality for $i=1$. 

An important consequence of this behavior is that, in contrast to simple non-coherent joint transmission, the outage rate is (asymptotically) larger or equal than the ergodic rate achieved by transmitter selection. This property can be easily formalized as follows:
\begin{proposition}
Let $\bar{R}_{\mathrm{out}}= \sup_{r\in \stdset{R}}r\cdot\P(\bar{R}(\alpha)\geq r)$ be the outage rate asymptotically achievable by the proposed extension of non-coherent joint transmission. Then:
\begin{equation*}
\bar{R}_{\mathrm{out}} = \max_{i\in\{1,\ldots,L\}} \P(\alpha \geq i) \bar{R}(i).
\end{equation*}
\end{proposition}
\begin{proof} 
Same steps as in the proof of Proposition~\ref{prop:C_out}.
\end{proof}

\begin{figure*}
\centerline{\subfigure{\includegraphics[width=2.5in]{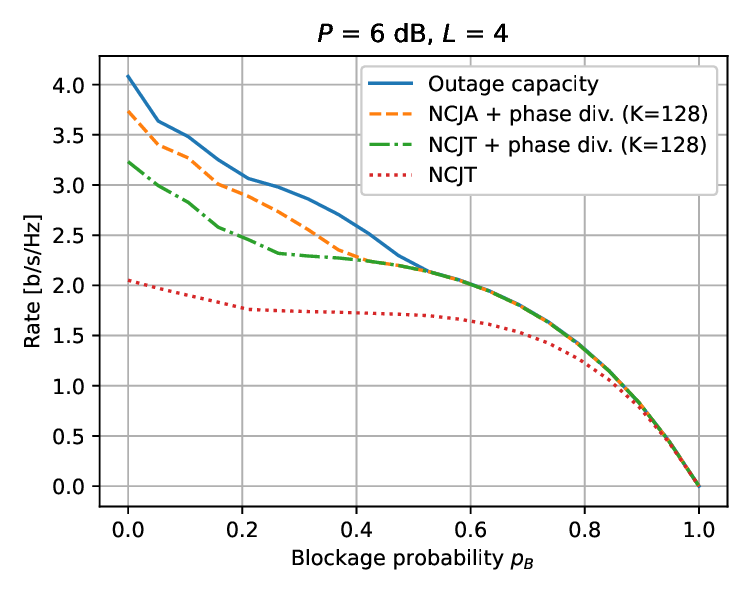}
\label{fig:R_out_vs_pB}}
\hfil
\subfigure{\includegraphics[width=2.5in]{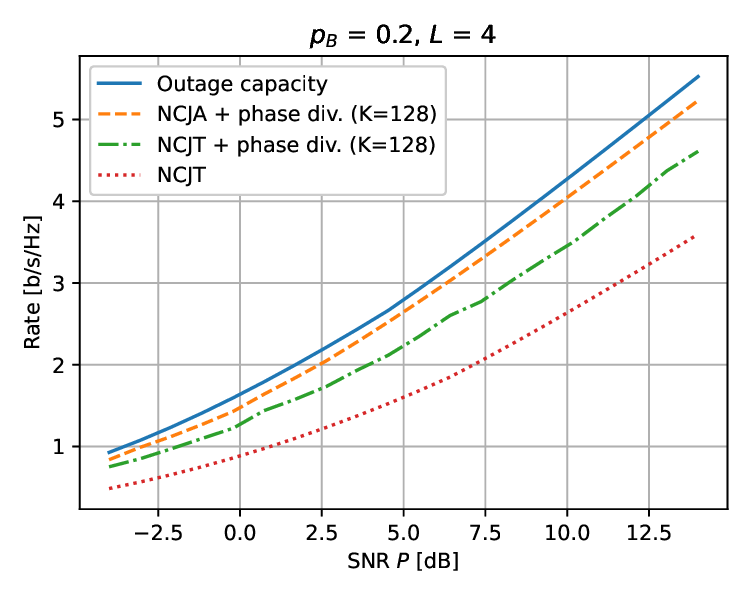}
\label{fig:R_out_vs_SNR}}
\hfil
\subfigure{\includegraphics[width=2.5in]{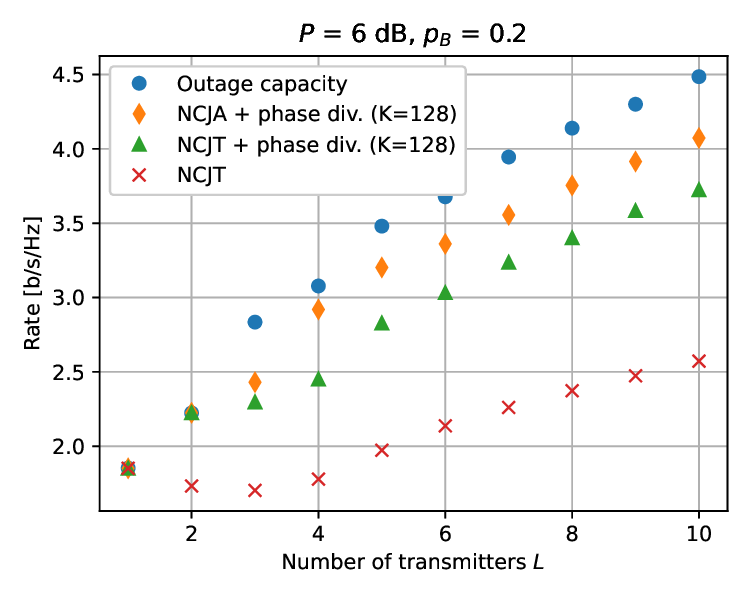}
\label{fig:R_out_vs_L}}}
\caption{Outage rate achieved by non-coherent joint transmission (NCJT) without phase diversity, NCJT with phase diversity, non-coherent joint Alamouti space-time coding (NCJA) with phase diversity, and capacity achieving schemes, by varying in each plot one of the three parameters $(p_B,P,L)$ of the simplified channel model in Section~\ref{sec:model}.} 
\label{fig:R_out}
\end{figure*}

\begin{corollary}
The following inequality holds:
\begin{equation*}
\bar{R}_{\mathrm{out}} \geq \P(\alpha \geq 1) \bar{R}(1) = (1-p_B^L)\log(1+P). 
\end{equation*}
\end{corollary}
Despite the asymptotic nature of the above property, our numerical results in Figure~\ref{fig:R_out} and Figure~\ref{fig:R_erg} show that similar conclusions may hold also for moderate values of $K$. Therefore, at the price of a slight increase in channel coding complexity (comparable to OFDM systems), the proposed extension can achieve the same or better effective rates as transmitter selection, without requiring fast network coordination and rate adaptation mechanisms based on CSIT. 

In terms of ergodic rates, the proposed extension achieves
\begin{equation*}
R = \E\left[\dfrac{1}{K}\sum_{k=1}^K\log(1+|h[k]|^2P)\right] = \E[\log(1+|h|^2P)],
\end{equation*}
that is, it does not provide any performance gain with respect to simple non-coherent joint transmission. However, the proposed extension is still very useful for the ergodic regime, since it significantly reduces the required complexity for approaching the ergodic rate via rate adaptation. Essentially, for large enough $K$, the encoder needs only to choose one out of $L$ coding rates, based on causal knowledge of the realizations $\alpha$ of the aggregate blockage process, as for the capacity achieving scheme. Similarly, the reduced fluctuations may significantly simplify the design of HARQ schemes. 

\begin{remark}
The proposed extension is based on the so-called phase diversity scheme \cite{dammann2002low}, which can be interpreted as a generalization of other known forms of diversity. In particular, the specific realization $(\forall l \in\{1,\ldots,L\})$ $(\forall k\in \{1,\ldots,K\})$ $\phi_{l,k} = 2\pi (k-1)d_l/K$ for some $d_l \in \stdset{N}$ gives the same performance and grid model in \eqref{eq:model_phase_div} as the so-called cyclic delay diversity scheme \cite{dammann2002low}, which is implemented by applying at each $l$th transmitter a cyclic shift $d_l$ to each frame of $K$ symbols, and by performing frequency-domain processing. This effectively corresponds to introducing frequency diversity, although the original channel is not frequency selective. We remark that this form of diversity have been mostly studied under classical multi-antenna fading models with rich scattering, such as the i.i.d. Rayleigh fading model, for which the effect of small-scale fading cannot be significantly mitigated by means of phase rotations at each antenna, as in this study.   
\end{remark}

\begin{figure*}
\centerline{\subfigure{\includegraphics[width=2.5in]{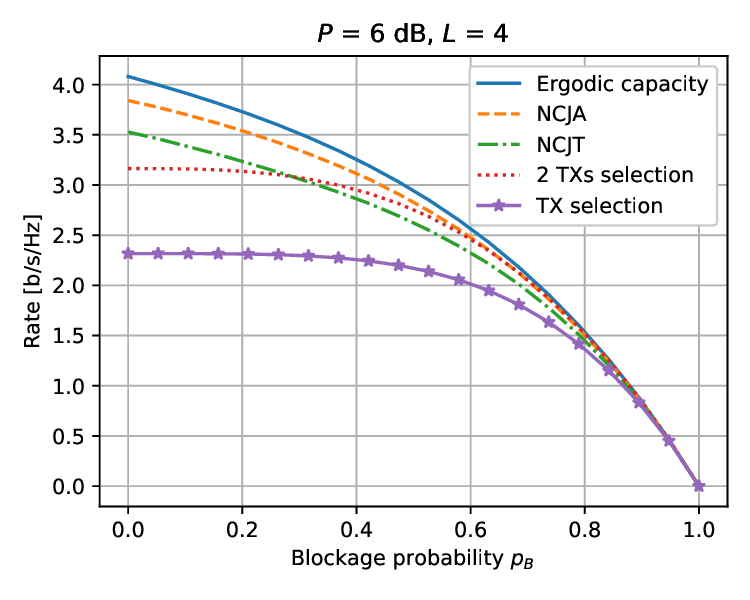}
\label{fig:R_erg_vs_pB}}
\hfil
\subfigure{\includegraphics[width=2.5in]{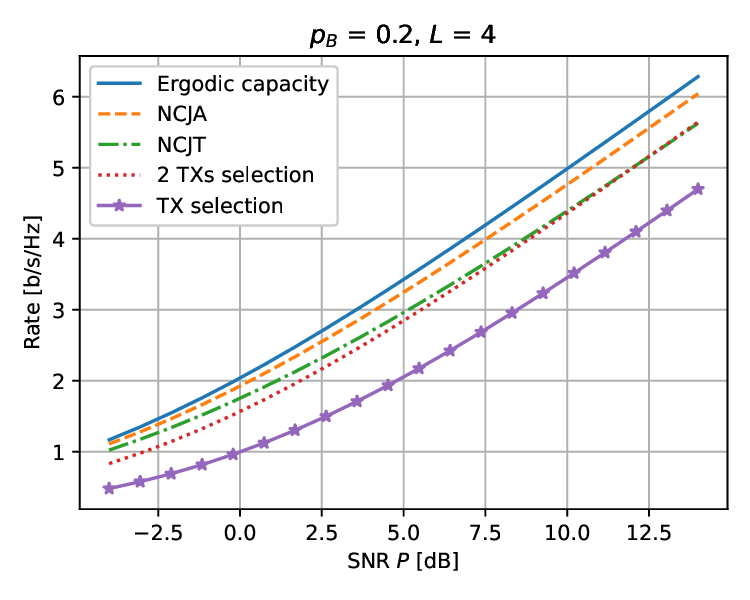}
\label{fig:R_erg_vs_SNR}}
\hfil
\subfigure{\includegraphics[width=2.5in]{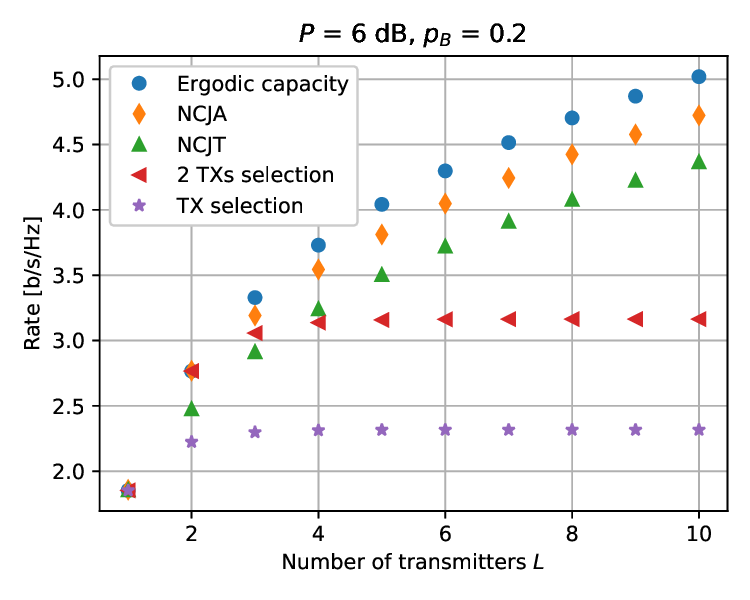}
\label{fig:R_erg_vs_L}}}
\caption{Ergodic rate achieved by transmitter selection, non-coherent joint transmission (NCJT), two transmitters selection, non-coherent joint Alamouti space-time coding (NCJA), and capacity achieving schemes, by varying in each plot one of the three parameters $(p_B,P,L)$ of the simplified channel model in Section~\ref{sec:model}. NCJT and NCJA can be with or without phase diversity.}
\label{fig:R_erg}
\end{figure*}

\subsection{Space-time coding}
For the special case of $L=2$ transmitters, the well-known Alamouti space-time block coding scheme \cite{alamouti1998} achieves both the ergodic and outage capacity by using simple linear receiver processing and scalar decoding, since it converts the original channel model into the effective single-input single-output model $(\forall m \in \stdset{Z})$
\begin{align*}
y[m] 
&= \sqrt{\alpha_t}u[m]+z[m], \quad t=\left\lfloor \frac{m}{T}\right\rfloor,
\end{align*}
where $u[m]$ is a scalar information bearing signal, which we assume i.i.d. $\CN(0,P)$. Unfortunately, the remarkable performance and simplicity of the Alamouti scheme cannot be extended to $L>2$ transmitters \cite{tarok1999space}. However, space-time coding schemes for $L>2$ transmitters with reasonable complexity-performance trade-off are still worthy of investigation. One possibility would be to revisit the available studies on some special classes of high-rate low-complexity space-time codes, such as quasi-orthogonal space-time block codes \cite{papadias2003capacity}, linear dispersion codes \cite{hassibi2002linear}, or similar alternatives, in light of the intermittent block fading model considered in this study. However, we leave this line of research for future work, and focus on the following simple enhancements of the transmitter selection and phase diversity schemes by means of an Alamouti space-time coding stage.

\subsubsection{Two transmitters selection} For each $t$th block-fading realization with at least one non-blocked transmitters, i.e., such that $\vec{\beta}_t\neq \vec{0}$, the network chooses a pair of transmitters $(l^{\star},l'^{\star})(\vec{\beta}_t)$ such that $(\beta_{l^\star,t},\beta_{l'^\star,t})\neq (0,0)$, and let them transmit the same scalar information bearing signal using Alamouti space-time coding. The other transmitters remain silent, as in the transmitter selection scheme. Similar to the case of $L=2$ transmitters, linear receiver processing converts the original channel model into the effective single-input single-output channel model ($\forall m \in \stdset{Z}$)
\begin{align*}
y[m] &= \sqrt{\min(2,\alpha_t)}u[m]+z[m], \quad t=\left\lfloor \frac{m}{T}\right\rfloor,
\end{align*}
where $u[m]$ is a scalar information bearing signal, which we assume i.i.d. $\CN(0,P)$. The following ergodic rate is achievable using standard techniques: 
\begin{align*}
R &= \P(\alpha \geq 2)\log(1+2P) + \P(\alpha = 1)\log(1+P). 
\end{align*} 
This scheme always outperforms transmitter selection, since it exploits the SNR gains offered by simultaneous transmission from two transmitters. In contrast, the comparison against non-coherent joint transmission is nontrivial and highly dependent on the system parameters, as also shown in Figure~\ref{fig:R_erg}.

\subsubsection{Non-coherent joint Alamouti space-time coding with phase diversity}
We consider the same grid channel model in \eqref{eq:model_phase_div}, and cluster the transmitters in pairs. For each pair of transmitters, we modify the transmit signals in \eqref{eq:signal_phase_div} by replacing the i.i.d information bearing signal $u[m,k]$ with its Alamouti space-time coded version (spanning two consecutive $K$-dimensional frames indexed by $m$). Linear receiver processing leads to the effective single-input single-output channel model $(\forall m \in \stdset{Z})$ $(\forall k\in \{1,\ldots,K\})$
\begin{equation*}
y[m,k] = \|\vec{h}_t[k]\|u[m,k] + z[m,k], \quad t=\left\lfloor \frac{m}{M}\right\rfloor,
\end{equation*}
where $u[m,k]\sim \CN(0,P)$, and $(\forall t\in \stdset{Z})$ $(\forall k\in \{1,\ldots,K\})$,
\begin{equation*}
\vec{h}_t[k]\eqdef \begin{bmatrix}
\sum_{l=1}^{L/2}\beta_{l,t}e^{j(\theta_{l,t}+\phi_{l,k})}\\
\sum_{l=L/2+1}^{L}\beta_{l,t}e^{j(\theta_{l,t}+\phi_{l,k})}
\end{bmatrix}.
\end{equation*}
The above scheme achieves the outage rate 
\begin{equation*}
R_{\mathrm{out}} = \sup_{r\in \stdset{R}} r\cdot \P\left(\dfrac{1}{K}\sum_{k=1}^K\log(1+\|\vec{h}[k]\|^2P) \geq r\right),
\end{equation*}
and the ergodic rate
\begin{equation*}
R = \E\left[\dfrac{1}{K}\sum_{k=1}^K\log(1+\|\vec{h}[k]\|^2P)\right]=\E[\log(1+\|\vec{h}\|^2P)],
\end{equation*}
where $\vec{h} \eqdef [
\sum_{l=1}^{L/2}\beta_l e^{j\theta_l}\quad 
\sum_{l=L/2+1}^{L}\beta_le^{j\theta_l}]^\T$ is the effective channel without phase diversity. Similar to non-coherent joint transmission, the use of phase diversity is motivated by the significant benefits it offers in terms of outage rate and simplification of rate adaptation / HARQ mechanisms for the ergodic regime. In particular, similar to non-coherent joint transmission, as $K$ grows large, the instantaneous rate becomes essentially driven by the blockage process: 
\begin{proposition}
Let $(\forall (i_1,i_2) \in \{0,\ldots,L/2\}^2)$ $\bar{R}(i_1,i_2)\eqdef $
\begin{equation*}
\begin{split}
\E\left[\log\left(1+\left| \textstyle \sum_{l=1}^{i_1}e^{j\theta_l}\right|^2P+\left| \textstyle \sum_{l=L/2+1}^{L/2+i_2}e^{j\theta_l}\right|^2P\right)\right].
\end{split}
\end{equation*}
Then, $(\forall \epsilon > 0)$
\begin{align*}
&\P\left(\left|\dfrac{1}{K}\sum_{k=1}^{K}\log(1+\|\vec{h}[k]\|^2P)-\bar{R}(\alpha_1,\alpha_2) \right|\geq \epsilon\right)\underset{K\to \infty}{\to} 0,
\end{align*}
where  $(\alpha_1,\alpha_2) \eqdef \left(\sum_{l=1}^{L/2}\beta_l,\sum_{l=L/2+1}^{L}\beta_l\right)$.
\end{proposition}
\begin{proof}
(Sketch) The proof is based on the weak law of large numbers and it is similar to the proof of Proposition~\ref{prop:lln}.
\end{proof}
Interestingly, Figure~\ref{fig:R_out} and Figure~\ref{fig:R_erg} show that the above scheme is able to recover a significant fraction of both the outage and the ergodic capacity in most regimes of interest. 

\section{Coarse time synchronization}
\label{sec:time}
We now assume that the network can jointly compensate the different physical propagation delays
of the signals from different transmitters only up to some maximum timing offset $\tau_{\max}>0$. As customary, we will use OFDM to provide robustness against unknown residual delays. In particular, we consider the time-domain channel model $(\forall m \in \stdset{Z})$
\begin{equation*}
y[m] = \sum_{l=1}^L(h_l * x_l)[m]+z[m],
\end{equation*}
where the impulse responses $h_l[m]$ for $l\in \{1,\ldots,L\}$ model the intermittent block fading process and the inter-symbol interference originating from the residual delays. Specifically, we let:
\begin{equation*}
h_l[m] = h_{l,t}g(m-\tau_l), \quad h_{l,t} = \beta_{l,t}e^{j\theta_{l,t}}, \quad t = \left\lfloor\frac{m}{T}\right\rfloor,
\end{equation*}
where $g$ is some continuous-time impulse response which models the convolution between the transmit and receive filters, and $\tau_l \leq \tau_{\max}$ is an unknown and possibly non-integer residual delay for the $l$th transmitter. To simplify the analysis, we assume square pulses at the transmitter side and matched filtering at the receiver side (similar to \cite{buzzi2018single}), i.e., a triangular impulse response
\begin{equation*}
g(x)= \begin{cases}
1-|x|, & x\in (-1,1) \\
0, & \text{otherwise},
\end{cases}
\end{equation*}
and neglect practical spectral mask constraints. We leave the analysis of more realistic and potentially better performing waveforms to future work.  


Assuming an OFDM system with $K$ subcarriers and a cyclic prefix of length $D \geq \lceil \tau_{\max} \rceil +1$, such that an integer number $M = T/(K+D)$ of OFDM symbols are transmitted in each fading block, the channel model in the time-frequency domain is given by $(\forall m \in \stdset{Z})$ $(\forall k\in \{0,\ldots,K-1\})$
\begin{equation*}
Y[m,k] = \sum_{l=1}^LH_{l,t}[k]X_l[m,k] + Z[m,k],\quad t=\left\lfloor \frac{m}{M}\right\rfloor, 
\end{equation*}
where $Z[m,k]\sim \CN(0,1)$ is a sample of a white Gaussian noise process (in both time and frequency), and where $H_{l,t}[k]$ is the $K$-points discrete-time Fourier transform (DFT) of $h_{l,t}[m]$. By splitting the delay $\tau_l$ into an integer part $d_l \in \stdset{N}$ and a fractional part $\delta_l \in [0,1)$ such that $\tau_l = d_l + \delta_l$, we obtain
\begin{equation*}
H_{l,t}[k] = \beta_{l,t} e^{j\theta_{l,t}}e^{-j2\pi \frac{k}{K}d_l}G_l[k],
\end{equation*}
where $G_l[k] \eqdef (1-\delta_l) + \delta_le^{-j2\pi\frac{k}{K}}$ is the DFT of $g(m-\delta_l)$.

\subsection{Capacity}
To capture the impact of unkown and uncontrollable residual delays, we follow a worst-case approach similar to the information theoretical literature on asynchronous \cite{cover1982asynchronous} or arbitrarily varying \cite{blackwell1960capacities} channels, and define ergodic capacity as the maximum achievable ergodic rate for all possible delays $\vec{\tau}= (\tau_1,\ldots,\tau_L) \in [0,\tau_{\max}]^L$. This implies that an upper bound on the ergodic capacity of the considered channel is readily given by the ergodic capacity under perfect time synchronization $C = \E[\log(1+\alpha P)]$, studied in Proposition~\ref{prop:C}. 

In addition, by applying for each subcarrier the signaling scheme achieving $C$, we obtain the following lower bound on the ergodic capacity:
\begin{align}\label{eq:C_worstcase_def}
R &= \inf_{\vec{\tau}\in[0,\tau_{\max}]^L}\dfrac{1}{K+D}\sum_{k=0}^{K-1}\E\left[\log\left(1+\sum_{l=1}^L|H_l[k]|^2P\right)\right]
\end{align} 
where $H_l[k] \eqdef \beta_le^{j(\theta_l-2\pi \frac{k}{K}d_l)}G_l[k]$ denotes a realization of $H_{l,t}[k]$. Since $|H_l[k]|^2=\beta_l|G_l[k]|^2$ does not depend on $d_l$, we notice that the integer parts $\vec{d}=(d_1,\ldots,d_L)$ of the delays $\vec{\tau}$ contribute to capacity loss only through the cyclic prefix length $D$, i.e., the system is insensitive to their actual values. On the other hand, the fractional parts $\vec{\delta}=(\delta_1,\ldots,\delta_L)$ contribute to capacity loss in a non-trivial manner through the induced frequency selectivity, i.e., through the fluctuations of $|G_l[k]|^2$ across the subcarriers. 

It turns out that the ergodic rate $R$ in \eqref{eq:C_worstcase_def} admits a closed form expression given by the intuitive case of completely off-grid sampling $(\forall l \in \{1,\ldots,L\})$~$\delta_l=0.5$. 
\begin{proposition}\label{prop:C_worstcase}
Consider the achievable ergodic rate~\eqref{eq:C_worstcase_def}. The following equality holds:
\begin{align*}
R &= \dfrac{1}{K+D}\sum_{k=0}^{K-1}\E\left[\log\left(1+\alpha \frac{P}{2}\left(1+\cos\left(\frac{2\pi k}{K}\right)\right)\right)\right].
\end{align*}
Furthermore, 
\begin{align*}
\lim_{K\to \infty} R &= \E\left[\log\left(1+\alpha \frac{P}{4}+\frac{1}{2}(\sqrt{1+\alpha P}-1)\right)\right].
\end{align*}
\end{proposition}
\begin{proof}
The proof is given in Appendix~\ref{proof:C_worstcase}.
\end{proof}

Following a similar worst-case approach, the outage capacity of the considered channel can be upper bounded by the outage capacity under perfect time synchronization 
\begin{equation*}
C_{\mathrm{out}}=\max_{i\in\{1,\ldots,L\}} \P(\alpha \geq i) \log(1+iP),
\end{equation*}
and lower bounded by
\begin{equation}\label{eq:C_out_worstcase_def}
\begin{split}
&R_{\mathrm{out}}=\inf_{\vec{\tau}\in[0,\tau_{\max}]^L} \sup_{r\in \stdset{R}} \\
&r\cdot\P\left(\dfrac{1}{K+D}\sum_{k=0}^{K-1}\log\left(1+\sum_{l=1}^L|H_l[k]|^2P\right)\geq r\right),
\end{split}
\end{equation}
which can be characterized in closed form as stated next.
\begin{proposition}\label{prop:C_out_worstcase}
Consider the achievable outage rate~\eqref{eq:C_out_worstcase_def}. The following equality holds:
\begin{equation*}
\begin{split}
&R_{\mathrm{out}}= \max_{i\in\{1,\ldots,L\}} \\
& \P(\alpha \geq i)\dfrac{1}{K+D}\sum_{k=0}^{K-1}\log\left(1+i \frac{P}{2}\left(1+\cos\left(\frac{2\pi k}{K}\right)\right)\right).
\end{split}
\end{equation*}
Furthermore, $\lim_{K\to \infty} R_{\mathrm{out}} = $
\begin{equation*}
\max_{i\in\{1,\ldots,L\}}\P(\alpha \geq i)\log\left(1+i\frac{P}{4}+\frac{1}{2}(\sqrt{1+i P}-1)\right).
\end{equation*}
\end{proposition}
\begin{proof}
(Sketch) The proof follows from the same techniques as in the proof of Proposition~\ref{prop:C_worstcase}.
\end{proof}

Proposition~\ref{prop:C_worstcase} and Proposition~\ref{prop:C_out_worstcase} show that the capacity of the considered system follows similar trends as in the synchronous case. In addition, they characterize the price of enforcing robustness against uncontrollable residual delays. More specifically, the above results show that  robust transmission can be theoretically achieved at the price of a  cyclic prefix overhead $\approx K/(K+D)$, which depends on the maximum uncontrollable integer delay, and a SNR loss factor $0.5+0.5 \cos(2\pi k /K)$ for each subcarrier $k$, which stems from the uncontrollable fractional delays. Clearly, the price of the multiplicative overhead can be mitigated if the number of subcarriers $K$ can be made sufficiently large, i.e., such that $D \ll K \ll T$ holds. Moreover, the above results show that, for large $K$, the effective SNR for a given number of non-blocked transmitters $\alpha\in \stdset{N}$ approaches 
\begin{equation*}
\alpha \frac{P}{4}+\frac{1}{2}(\sqrt{1+\alpha P}-1) \geq \alpha \frac{P}{4},
\end{equation*}
i.e., the losses can be limited to $6$dB from the ideal SNR $P$.

\subsection{Non-coherent joint transmission and phase diversity} 
As a more practical alternative to the near-optimal scheme studied in the previous section, we now revisit non-coherent joint transmission, i.e., the transmission of a single scalar codeword from all transmitters simultaneously, with and without phase diversity. By taking the same worst-case approach as in the capacity analysis, we first consider the achievable ergodic rate (with or without phase diversity) $R=$
\begin{align}\label{eq:R_worstcase_def}
\inf_{\vec{\tau}\in[0,\tau_{\max}]^L}\dfrac{1}{K+D}\sum_{k=0}^{K-1}\E\left[\log\left(1+\left|\sum_{l=1}^LH_l[k]\right|^2P\right)\right].
\end{align}
The above expression can be evaluated as stated next. Similarly to the capacity analysis, the resulting expression corresponds to the case of completely off-grid sampling.
\begin{proposition}\label{prop:R_worstcase}
Consider the achievable ergodic rate~\eqref{eq:R_worstcase_def}. The following equality holds: 
\begin{align*}
R = \dfrac{1}{K+D}\sum_{k=0}^{K-1}\E\left[\log\left(1+\left|h\right|^2 \frac{P}{2}\left(1+\cos\left(\frac{2\pi k}{K}\right)\right)\right)\right],
\end{align*}
where $h\eqdef \sum_{l=1}^ L\beta_le^{j\theta_l}$.
Furthermore, 
\begin{align*}
\lim_{K\to \infty} R = \E\left[\log\left(1+|h|^2 \frac{P}{4}+\frac{1}{2}(\sqrt{1+|h|^2 P}-1)\right)\right].
\end{align*}
\end{proposition}
\begin{proof}
The proof is given in Appendix~\ref{proof:R_worstcase}. 
\end{proof}
The above proposition shows that the ergodic rate achieved by non-coherent joint transmission follows similar trends as in the synchronous case, and experience similar penalties due to cyclic prefix overhead and off-grid sampling as in the worst-case capacity analysis. Unfortunately, the exact (worst-case) achievable outage rate seems difficult to characterize for finite $K$. We know that, if phase diversity is not used, the outage performance must be at least as bad as in the synchronous case and include the aforementioned penalties to ensure robustness against uncontrollable delays. Phase diversity is expected to improve performance significantly, but we provide a formal justification only in an asymptotic sense. 

Building on the OFDM grid model, we apply the phase diversity technique directly in the frequency domain, i.e., we let $(\forall m \in \stdset{Z})$ $(\forall k\in \{0,\ldots,K-1\})$ $(\forall l \in \{1,\ldots,L\})$
\begin{equation*}
X_l[m,k] = e^{j\phi_{l,k}}U[m,k],
\end{equation*}
where $U[m,k]\sim \CN(0,P)$ is a scalar i.i.d. information bearing signal, and where $(\forall l \in\{1,\ldots,L\})$ $(\phi_{l,0},\ldots,\phi_{l,K-1})$ is a vector of random phases independently and uniformly distributed in $[0,2\pi]$. For given delays $\vec{\tau}$, the following instantaneous rate is achievable on subcarrier
$k\in \{0,\ldots,K-1\}$:
\begin{align}\label{eq:R_inst_asynch}
R_k= \log\left(1+\left|\sum_{l=1}^LH_l[k]e^{j\phi_{l,k}}\right|^2P\right).
\end{align}
Similar to the synchronous case, we then observe that phase diversity makes the instantaneous rate fluctuations $\frac{1}{K+D}\sum_{k=0}^{K-1}R_k$ essentially driven by the blockage process.
\begin{proposition}\label{prop:hoeffding}
Consider the per-subcarrier instantaneous rate in \eqref{eq:R_inst_asynch}, and let $(\forall \vec{b}\in \{0,1\}^L)$ $(\forall k \in \{0,\ldots,K-1\})$
\begin{equation*}
\bar{R}_k(\vec{b}) \eqdef \E\left[\log\left(1+\left|\textstyle\sum_{l=1}^Lb_le^{j\theta_l}G_l[k]\right|^2P\right)\right].
\end{equation*}
Then $(\forall \epsilon > 0)$
\begin{align*}
\lim_{K\to \infty}\P\left(\left|\dfrac{1}{K+D}\sum_{k=0}^{K-1}\left(R_k-\bar{R}_k(\vec{\beta})\right) \right|\geq \epsilon\right)=0,
\end{align*}
i.e., the instantaneous rate converges in probability to $\frac{1}{K+D}\sum_{k=0}^{K-1}\bar{R}_k(\vec{\beta})$ as $K$ grows large.
\end{proposition}
\begin{proof}
The proof is based on Hoeffding's inequality. The details are given in Appendix~\ref{proof:hoeffding}.
\end{proof}
On top of simplifying practical implementations of rate adaptation and HARQ mechanisms for the ergodic regime, the above property allows us to approximately evaluate the (worst-case) outage performance by focusing on the asymptotically achievable outage rate
\begin{equation}\label{eq:R_out_worstcase_def}
\bar{R}_{\mathrm{out}}= \inf_{\vec{\tau}\in[0,\tau_{\max}]^L} \sup_{r\in \stdset{R}} r\cdot\P\left(\dfrac{1}{K+D}\sum_{k=0}^{K-1}\bar{R}_k(\vec{\beta})\geq r\right),
\end{equation}
which can be characterized as stated next.
\begin{proposition}\label{prop:R_out_worstcase}
Consider the achievable outage rate~\eqref{eq:R_out_worstcase_def}. The following equality holds:
\begin{equation*}
\begin{split}
\bar{R}_{\mathrm{out}}= \max_{i\in\{1,\ldots,L\}} \P(\alpha \geq i)\dfrac{1}{K+D}\sum_{k=0}^{K-1}\bar{R}_k'(i),
\end{split}
\end{equation*}
where $(\forall k \in \{0,\ldots,K-1\})$ $(\forall i \in\{0,\ldots,L\})$ $\bar{R}_k'(i) \eqdef $
\begin{equation*}
\E\left[\log\left(1+\left|\textstyle\sum_{l=1}^ie^{j\theta_l}\right|^2\frac{P}{2}\left(1+\cos\left(\frac{2\pi k}{K}\right)\right)\right)\right].
\end{equation*}
Furthermore, 
\begin{equation*}
\lim_{K\to \infty} \bar{R}_{\mathrm{out}} = \max_{i\in\{1,\ldots,L\}}\P(\alpha \geq i)\bar{R}(i),
\end{equation*}
where $(\forall i\in \{0,\ldots,L\})$ $\bar{R}(i) \eqdef $
\begin{equation*}
\E\left[\log\left(1+\left|\sum_{l=1}^ie^{j\theta_l}\right|^2 \frac{P}{4}+\frac{\sqrt{1+\left|\sum_{l=1}^ie^{j\theta_l}\right|^2 P}-1}{2}\right)\right].
\end{equation*}
\end{proposition}
\begin{proof}
(Sketch) The proof follows from the same techniques as in the proof of Proposition~\ref{prop:R_worstcase}.
\end{proof}
As expected, we observe that the asymptotically achievable outage rate follows similar trends as in the synchronous case, and experience similar penalties as in the worst-case capacity analysis. We conclude this section by pointing out that the above discussion can be extended to a related frequency-domain version of the non-coherent joint Alamouti scheme with phase diversity studied in Section~\ref{sec:schemes}. However, we omit the details due to space limitation. For the same reason, we do not repeat Figure~\ref{fig:R_out} and Figure~\ref{fig:R_erg} for the asynchronous case. Nevertheless, we report that, as expected, the curves are similar to the synchronous case up to a SNR penalty of roughly 6dB.

\begin{figure}[h!]
\centering
\begin{overpic}[abs,unit=1mm,scale=.25,width=2.3in]{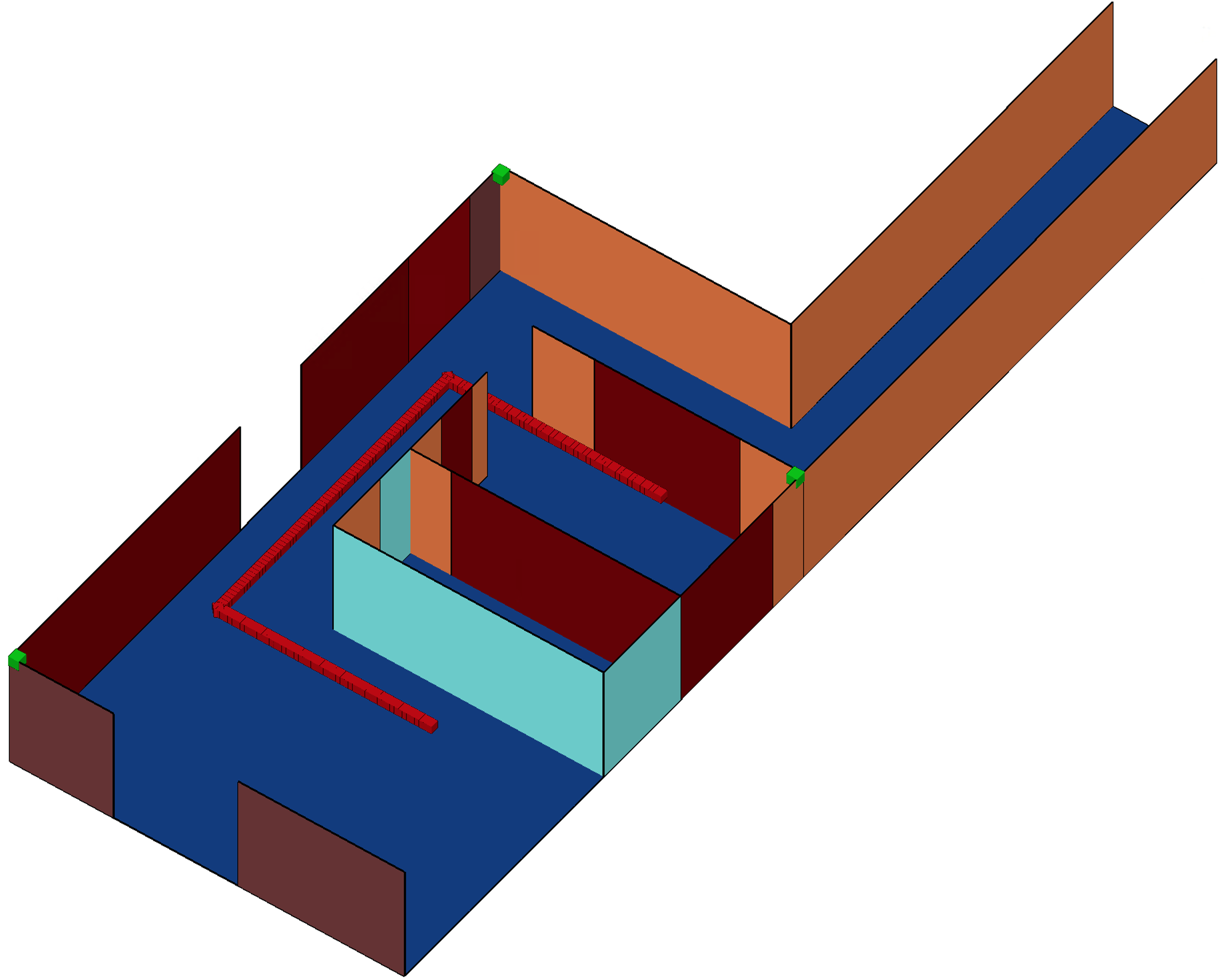}
\put(0,16){\color{green} TX~1}
\put(25,37){\color{green} TX~2}
\put(37,20){\color{green} TX~3}
\put(15,9){\color{red} RX}
\end{overpic}
\caption{Simulated indoor environment with $L=3$ transmitters jointly serving a single receiver moving along a predefined trajectory.}
\label{fig:setup}
\end{figure}
\section{Ray-tracing simulations}
\label{sec:sim}
In this section we validate the effectiveness of our theoretical analysis by means of realistic ray-tracing simulations. We use Remcom's Wireless InSite® to construct an indoor environment and to generate samples of the channel between $L=3$ transmitters and a receiver moving along a trajectory, as shown in Figure~\ref{fig:setup}. Importantly, in contrast to the simplified statistical model in Section~\ref{sec:model}, the channel samples are formed based on a deterministic model that does not neglect the impact of multipath propagation, blocked signals, and different path loss from different transmitters. Nevertheless, we stress that, due to the inherent deterministic nature and high computational complexity of ray tracing simulations, the goal of this section is not to perform a comprehensive statistical assessment of the large-scale system behavior by varying parameters such as the probability of blockage, as depicted in Figure~\ref{fig:R_out} and Figure~\ref{fig:R_erg}. Rather, the goal of this section is to verify the effectiveness of the considered transmission techniques in a hypothetical yet realistic system deployment. For example, in this section, we do not investigate the average system performance under rapid blockage events (such as human blockage events), since none of the existing ray-tracing simulators support accurate physical models of this complex effect. Nevertheless, this section investigates the effectiveness of the transmission techniques over multiple small-scale fading realizations, conditioned on specific instantaneous realizations of the blockage process. We remark that, as commonly done in deterministic ray tracing simulations, the impact of small-scale fading is characterized by replacing statistical operators with corresponding time averaging operators, constructed by finely sampling the channel around a given trajectory point (sub-wavelength spacing, i.e., fractions of millimeter). 

\subsection{Simulation setup}
We set the carrier frequency to $f_c = 100$~GHz, i.e., to the maximum recommended value within Wireless InSite®, and consider a large $10\%$ fractional bandwidth $B = 10$ GHz. We let each transmitter be equipped with a $20\times 20$ critically spaced uniform planar array, driven by power amplifiers with saturated output power $P_{\mathrm{sat}}= 10$~dBm (feasible with CMOS technology), resulting in a maximum equivalent isotropic output power $\mathrm{EIRP} \approx 36$~dBm. We further assume an output backoff $\mathrm{OBO}=6$dB, to mitigate OFDM signal distortion due to transmitter non-linearities without resorting to complex pre/post-distortion techniques. The receiver is omnidirectional, and has a relatively high noise figure of $F=10$ dB, to model the known difficulties in realizing high-quality hardware at these frequencies. The thermal noise spectral density is set to the typical value $N_0 = -175$~dBm/Hz. We assume line-of-sight beam alignment for the entire receiver trajectory, which in practice can be achieved using position side information. 

\subsection{Impact of multipath propagation and blocked signals}
\begin{figure}[h!]
\centering
\includegraphics[width=0.8\linewidth]{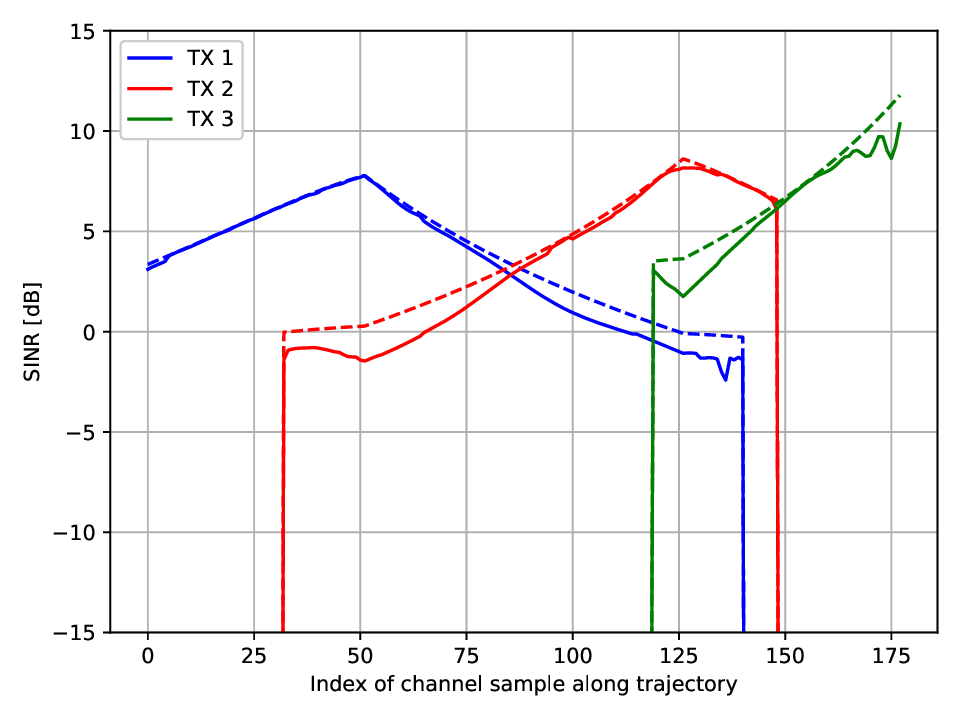}
\caption{Validation of negligible impact of multipath propagation and blocked signals over the entire receiver trajectory.}
\label{fig:singlepath}
\end{figure}
Figure~\ref{fig:singlepath} reports the per-transmitter SNR for equally spaced channel samples along the entire trajectory, calculated by taking the strongest path as useful signal term, and where the power of the other multipath components is either included in the noise term (solid line) or omitted (dashed line). As expected, we observe that line-of-sight blockage events induce extreme SNR losses. Furthermore, we observe that, when the line-of-sight is not blocked, the power of the multipath components is either comparable to the thermal noise power, or negligible. Therefore, Figure~\ref{fig:singlepath} corroborates our assumption in Section~\ref{sec:model} that the impact of multipath propagation and blocked signals can be safely neglected. This, in turn, supports the effectiveness of the intermittent block fading channel model introduced in Section~\ref{sec:model} to simplify the analysis.

\subsection{Impact of artificially induced small-scale fading}
\begin{figure}[h!]
\centering
\includegraphics[width=0.8\linewidth]{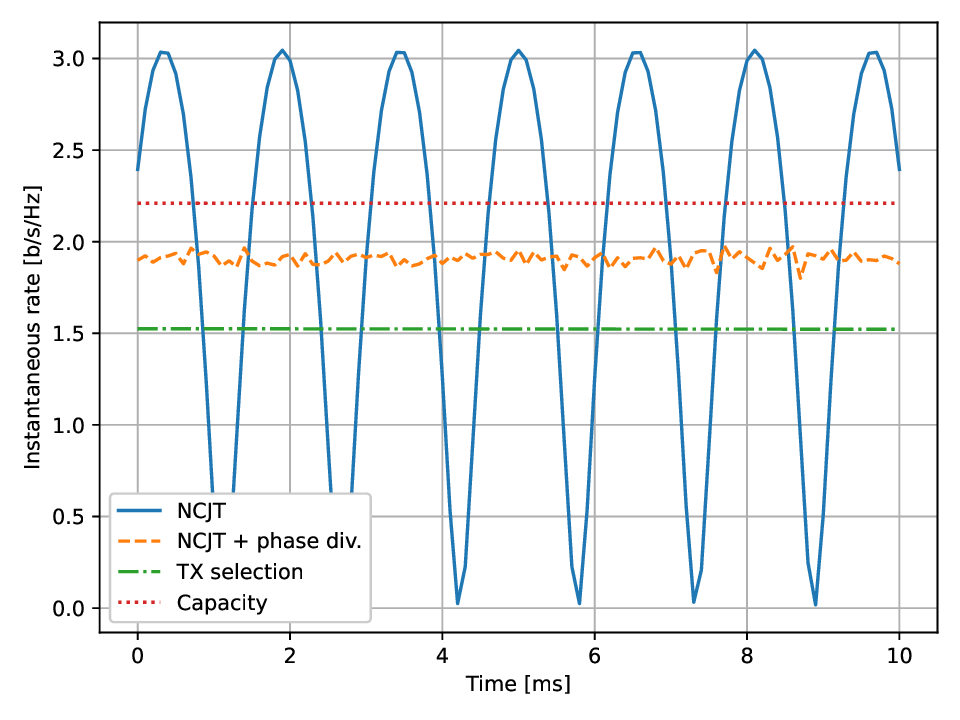}
\caption{Validation of predicted small-scale fading behaviour of different transmission schemes around a given trajectory point. The average rate $\bar{r}$ of NCJT is approximately $1.9$~b/s/Hz, with or without phased diversity.}
\label{fig:fading}
\end{figure}
Figure~\ref{fig:fading} reports the instantaneous achievable rate for $100$ (sub-wavelength spaced) channel samples in the center of the trajectory, by assuming a constant receiver speed of $v=1$~m/s. We focus on perfect synchronization for simplicity, but we remark that, as formally studied in Section~\ref{sec:time}, the worst-case performance under coarse synchronization follows similar trends. As predicted by our theoretical analysis, we observe that non-coherent joint transmission experiences significant signal fluctuations and deep fades. This is due to the small-scale fading artificially induced by the non-coherent superposition of the two strong line-of-sight paths from TX~1 and TX~2 in Figure~\ref{fig:setup}. In contrast, we observe that the use of transmitter selection, phase diversity, or capacity-achieving techniques makes the instantaneous rate much more stable over time, even in the presence of multipath propagation. Note that the performance of transmitter selection reported in Figure~\ref{fig:fading} also corresponds to the performance of all other transmission schemes when the line-of-sight path from either TX~1 or TX~2 is blocked. This further confirms the potential of transmitter selection, phase diversity, or capacity-achieving techniques in providing stable instantaneous rates over multiple blockage events, as the instantaneous rates vary over (essentially) discrete levels that depend only on the specific blockage state. Furthermore, although not explicitly reported, we stress that Figure~\ref{fig:fading} also implicitly shows the instantaneous rate of non-coherent joint Alamouti space-time coding with phase diversity. Depending on how the transmitters are clustered, this technique would yield the same instantaneous rate as either non-coherent joint transmission with phase diversity, or capacity achieving techniques. Hence, this technique will not be further discussed explicitly.

\subsection{Benefits of phase diversity}
\begin{figure}[h!]
\centering
\includegraphics[width=0.8\linewidth]{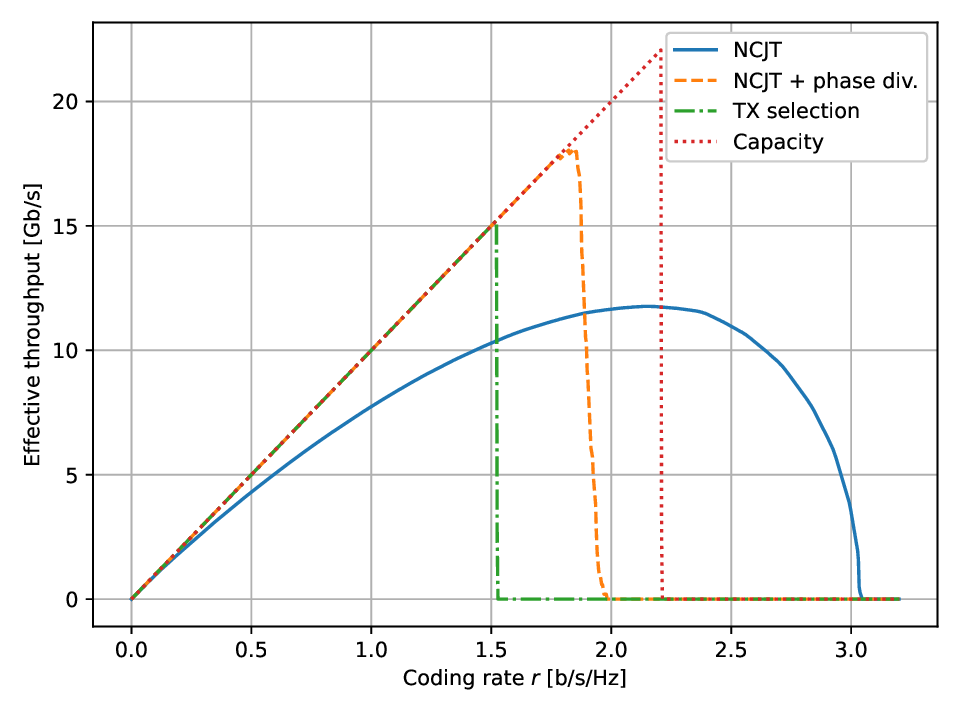}
\caption{Effective throughput for a fixed modulation and coding scheme of rate $r$, under the same setup as in Figure~\ref{fig:fading}.}
\label{fig:outage}
\end{figure}
The benefits of phase diversity are most noticeable by reading Figure~\ref{fig:fading} in terms of outage performance. For example, by assuming a fixed modulation and coding scheme of rate $r= 1.5$ b/s/Hz, non-coherent joint transmission would result into loosing approximately $1/3$ of the packets, hence delivering an effective throughput of $\frac{2}{3}Br \approx 10$ Gb/s, instead of the $Br \approx 15$ Gb/s throughput delivered by transmitter selection. Perhaps more importantly, we notice that non-coherent joint tranmission leads to bursts of data loss of duration $\approx 0.5$ ms every $1.5$ ms, corresponding to loosing $15$ Gb/s $\cdot~0.5$ ms $= 7.5$ Mb every $1.5$ ms. Clearly, these figures would put traditional retransmission protocols under exceptional stress, for instance in the context of buffer design and management. In contrast, phase diversity makes the performance much more stable and supports the same $15$~Gb/s throughput as transmitter selection, while keeping the key benefit of non-coherent joint transmission, i.e., reduced coordination requirements to counteract blockages. In fact, it can even support higher throughput (up to $18$~Gb/s) than transmitter selection, as more formally reported in Figure~\ref{fig:outage}. Furthermore, similar to all the considered suboptimal trasnmission schemes, it does not require complex receiver processing as for capacity achieving schemes based on vector-valued codewords.

In terms of average performance, phase diversity brings no direct gains, as predicted by theory. In particular, by computing the average $\bar{r}$ of the instantaneous rates in Figure~\ref{fig:fading}, we obtain a throughput of $B\bar{r} \approx 19$ Gb/s for non-coherent joint transmission either with or without phase diversity. However, without phase diversity, approaching this average throughput with practical rate adaptation mechanisms seems very challenging, as the system would need to react to continuous rate changes spanning a $3$~b/s/Hz dynamic range every $1.5$~ms. In contrast, phase diversity can approach this average throughput by adapting the rate on a large-scale basis only, i.e., based on the SNR of the non-blocked transmitters. Similar observations hold for the design of HARQ mechanisms.

\section{Conclusion}
\label{sec:conclusion}
Our results provide a set of practical design guidelines for realizing full macrodiversity gains and significant SNR gains in mmWave/sub-THz downlink access networks with low-complexity receivers and minimal coordination and synchronization requirements at the infrastructure side. In particular, we illuminate the importance of transmit diversity techniques to make non-coherent joint transmission a competitive alternative to standard transmitter selection. Specifically, the proposed extension based on phase diversity significantly improves the outage performance of non-adaptive (fixed-rate) coding schemes, and it allows to design rate adaptation or HARQ mechanisms focusing only on the large-scale fluctuations given by the relatively slow blockage process. Furthermore, we illuminate the potential of simple space-time coding schemes to recover a significant fraction of the optimal performance. Finally, we show that robustness against timing offsets can be achieved using OFDM and a simple SNR margin.
 
\appendix
\subsection{Proof of Proposition~\ref{prop:monotonicity}}
\label{app:monotonicity}
We first state two auxiliary lemmas.
\begin{lemma}\label{lem:closedform}
Let $\theta\sim \text{Uniform}(0,2\pi)$. Then $(\forall x\in \stdset{C})$ $(\forall y  \in \stdset{C})$
\begin{equation*} 
\E[\log(1+|x+e^{j\theta}y|^2)] = \log\left(\dfrac{1+a+\sqrt{(1+a)^2-b^2}}{2}\right),
\end{equation*}
where $a := |x|^2+|y|^2$ and $b := 2|x||y|$. 
\end{lemma}
\begin{proof}
We define the shorthand $c := \frac{b}{1+a} \in [0,1)$, and let:
\begin{align*}
&\E[\log(1+|x+e^{j\theta}y|^2)] \\
&= \E[\log(1+a+b\cos(\theta))]\\
&= \log(1+a) + \E[\log(1+c\cos(\theta))]\\
&= \log(1+a) + \frac{1}{\pi}\int_0^{\pi}\log\left(1+c\cos(\theta)\right)d\theta \\
&= \log\left(\dfrac{1+a+\sqrt{(1+a)^2-b^2}}{2}\right),
\end{align*} 
where the last equality follows from algebraic manipulations and the known definite integral
\begin{equation*}
\int_0^\pi\log\left(1+c\cos(\theta)\right)d\theta = \pi\log\left(\frac{1+\sqrt{1-c^2}}{2} \right), \quad |c|<1.
\end{equation*}
\end{proof}

\begin{lemma}\label{lem:monotone}For all $x\in \stdset{C}$, the function $f: \stdset{R}_+ \to \stdset{R}$ given by
\begin{equation*} 
f(|y|^2)\eqdef \log\left(\dfrac{1+a+\sqrt{(1+a)^2-b^2}}{2}\right),
\end{equation*}
where $a := |x|^2+|y|^2$ and $b := 2|x||y|$, is strictly increasing within its domain.
\end{lemma}
\begin{proof}(Sketch)
Standard calculus shows that the argument of the logarithm has strictly positive derivative for all $x\in \stdset{C}$.
\end{proof}

We are now ready to prove that $\bar{R}(i)$ is strictly increasing, i.e., that
$(\forall i \geq 0)~\bar{R}(i+1) > \bar{R}(i)$. We focus on $i\geq 1$, since the case $i=0$ is trivial. We define the random variable $X  \eqdef \sum_{l=1}^ie^{j\theta_l}$, independent of $\theta_{i+1}$, and observe that
\begin{align*}
\bar{R}(i+1) &=\E[\log(1+|X+e^{j\theta_{i+i}}|^2P)] \\
&= \E[\E[\log(1+|X+e^{j\theta_{i+1}}|^2P)|X]] \\
& > \E[\log(1+|X|^2P)|X]] = \bar{R}(i),
\end{align*}
where the inequality follows by Lemma~\ref{lem:closedform} and Lemma~\ref{lem:monotone}.

We finally show that $\bar{R}(i)$ is unbounded above, as a direct implication of the following inequality:
\begin{equation*}
\lim_{i\to \infty}\dfrac{\bar{R}(i)}{\log(1+iP)} \geq e^{-1}.
\end{equation*}
To prove the above inequality, we let $X_i \eqdef \frac{1}{\sqrt{i}}\sum_{l=1}^ie^{j\theta_l}$, and observe that by Markov's inequality  $(\forall i \geq 1)$
\begin{align*}
\dfrac{\bar{R}(i)}{\log(1+iP)} &\geq \P\left(\log(1+i|X_i|^2P)\geq \log(1+iP)\right) \\
&= \P(|X_i|^2\geq 1).
\end{align*}
By the central limit theorem, $X_i$ converges in distribution to $X\sim \CN(0,1)$ as $i\to \infty$. Hence, by the continuous mapping theorem, $|X_i|^2$ converges in distribution to $|X|^2\sim \text{Exp}(1)$. Taking limits on both sides concludes the proof: 
\begin{equation*}
\lim_{i\to\infty}\dfrac{\bar{R}(i)}{\log(1+iP)} \geq \lim_{i\to\infty}\P(|X_i|^2\geq 1) = \P(|X|^2\geq 1).
\end{equation*}

\subsection{Proof of Proposition~\ref{prop:R_NCJT}}
\label{app:R_NCJT}
We write $R(L)$ to highlight the dependency of $R$ on $L$. The proof relies on the following identity, which follows by the law of total expectation: $(\forall L\geq 1)$
\begin{align*}
R(L) &= \E\left[\E\left[\log\left(1+\left|\textstyle\sum_{l=1}^L\beta_le^{j\theta_l}\right|^2P\right)\middle|\beta_1,\ldots,\beta_L\right]\right] \\
&= \E[\bar{R}(\alpha_L)], 
\end{align*}
where $\alpha_L \eqdef \sum_{l=1}^L\beta_l$. To prove (i), we observe that $(\forall L\geq 1)$
\begin{align*}
\E[\bar{R}(\alpha_L)] &=\sum_{i = 1}^L\P(\alpha_L = i)\bar{R}(i)\\
&\geq \sum_{i=1}^L\P(\alpha_L=i)\bar{R}(1) = (1-p_B^L)\log(1+P), 
\end{align*}
where the inequality follows from Proposition~\ref{prop:monotonicity}. To prove (ii), we define the shorthand $\Delta\bar{R}(i)\eqdef \bar{R}(i)-\bar{R}(i-1)>0$, where the inequality follows from Proposition~\ref{prop:monotonicity}. We then observe that $(\forall L\geq 2)~\E[\bar{R}(\alpha_L)] = \sum_{i = 1}^L\P(\alpha_L \geq i)\Delta\bar{R}(i)$
\begin{align*}
&= \P(\alpha_L \geq L)\Delta\bar{R}(L)+\sum_{i = 1}^{L-1}\P(\alpha_L \geq i)\Delta\bar{R}(i)\\
&\geq \P(\alpha_L \geq L)\Delta\bar{R}(L)+\sum_{i=1}^{L-1}\P(\alpha_{L-1} \geq i)\Delta\bar{R}(i) \\
&= \P(\alpha_L \geq L)\Delta\bar{R}(L)+\E[\bar{R}(\alpha_{L-1})] >\E[\bar{R}(\alpha_{L-1})],
\end{align*}
where the first inequality follows from the property $(\forall L \geq 1)(\forall i \geq 0)$ $\P(\alpha_L\geq i) \geq \P(\alpha_{L-1} \geq i)$. 

To prove (iii), we observe that, by Markov's inequality:
\begin{equation*}
(\forall i \geq 0)~\E[\bar{R}(\alpha_L)] \geq \P(\alpha_L \geq i)\bar{R}(i).
\end{equation*}
Taking the limit on both sides gives
\begin{equation*}
(\forall i\geq 0)~\lim_{L\to \infty}\E[\bar{R}(\alpha_L)] \geq \lim_{L\to \infty}\P(\alpha_L\geq i)\bar{R}(i) = \bar{R}(i),
\end{equation*}
which, by Proposition~\ref{prop:monotonicity}, implies $\lim_{L\to \infty}\E[\bar{R}(\alpha_L)]= \infty$.

\subsection{Proof of Proposition~\ref{prop:lln}}
\label{proof:lln}
By the law of total probability, we have
\begin{align*}
&\P\left(\left|\dfrac{1}{K}\sum_{k=1}^{K}\log(1+|h[k]|^2P)-\bar{R}(\alpha) \right|\geq \epsilon\right)\\
&=\E\left[\P\left(\left|\dfrac{1}{K}\sum_{k=1}^{K}\log(1+|h[k]|^2P)-\bar{R}(\alpha) \right|\geq \epsilon \middle| \vec{\beta},\vec{\theta}\right) \right].
\end{align*}
We then observe that, conditioned on $(\vec{\beta},\vec{\theta})$, $\{\log(1+|h[k]|^2P)\}_{k=1}^K$ is an i.i.d. sequence with mean $\bar{R}(\alpha)$. Hence, the weak law of large numbers applies and 
\begin{equation*}
\P\left(\left|\dfrac{1}{K}\sum_{k=1}^{K}\log(1+|h[k]|^2P)-\bar{R}(\alpha) \right|\geq \epsilon \middle| \vec{\beta},\vec{\theta}\right) \underset{K\to \infty}{\to} 0.
\end{equation*}
The proof is concluded by exchanging limit and expectation based, e.g., on the dominated convergence theorem (probabilties are trivially bounded). 

\subsection{Proof of Proposition~\ref{prop:C_worstcase}}
\label{proof:C_worstcase}
For the first part of the statement, we  observe that
\begin{align*}
R &= \inf_{\vec{\tau}\in[0,\tau_{\max}]^L}\dfrac{1}{K+D}\sum_{k=0}^{K-1}\E\left[\log\left(1+\sum_{l=1}^L|H_l[k]|^2P\right)\right]\\
&= \inf_{\vec{\delta}\in[0,1)^L}\dfrac{1}{K+D}\sum_{k=0}^{K-1}\E\left[\log\left(1+\sum_{l=1}^L\beta_l|G_l[k]|^2P\right)\right].
\end{align*}
Then, we notice that $|G_l[k]|^2=
\delta_l^2 + (1-\delta_l)^2 + 2\delta_l(1-\delta_l)\cos(2\pi k/K)$ is an upward facing parabola in $\delta_l$, with minimum at $\delta_l = 0.5$, for all subcarriers $k\in \{0,\ldots,K-1\}$. Since the argument of the expectation is monotonic increasing in $|G_l[k]|^2$, the desired expression follows by letting $(\forall l \in \{1,\ldots,L\})$~$\delta_l=0.5$ and from simple algebraic manipulations. For the second part of the statement, we observe that, by standard properties of Riemman sums, we have $(\forall \alpha \in \stdset{N})$
\begin{align*}
&\lim_{K\to \infty} \frac{1}{K+D}\sum_{k=0}^{K-1} \log\left(1+\alpha \frac{P}{2}\left(1+\cos\left(\frac{2\pi k}{K}\right)\right)\right)\\
=&\frac{1}{2\pi}\int_{-\pi}^\pi\log\left(1+\alpha \frac{P}{2}\left(1+\cos\left(\theta\right)\right)\right)d\theta. 
\end{align*}
The integral can be evaluated in closed form using the same technique as in the proof of Lemma~\ref{lem:closedform} in Appendix~\ref{app:monotonicity}, concluding the proof. We remark that the limit and the expectation can be exchanged since $\alpha$ is a discrete random variable.

\subsection{Proof of Proposition~\ref{prop:hoeffding}}
\label{proof:hoeffding}
By the law of total probability, we have
\begin{align*}
&\P\left(\left|\dfrac{1}{K+D}\sum_{k=0}^{K-1}\left(R_k-\bar{R}_k(\vec{\beta})\right) \right|\geq \epsilon\right)\\
&=\E\left[\P\left(\left|\dfrac{1}{K+D}\sum_{k=0}^{K-1}\left(R_k-\bar{R}_k(\vec{\beta})\right) \right|\geq \epsilon \middle| \vec{\beta},\vec{\theta}\right) \right].
\end{align*}
We then observe that, conditioned on $(\vec{\beta},\vec{\theta})$, each $R_k$ is an independent random variable with mean $\bar{R}_k(\vec{\beta})$. Furthermore, we observe that $\frac{1}{K+D}R_k$ is readily upper bounded by $\frac{1}{K}\log(1+L^2P)$. Hence, we can apply Hoeffding's inequality on the sum of bounded independent random variables \cite{hoeffding1963}, which gives
\begin{align*}
\P\left(\left|\sum_{k=0}^{K-1}\frac{R_k-\bar{R}_k(\vec{\beta})}{K+D}\right|\geq \epsilon \middle| \vec{\beta},\vec{\theta}\right)\leq 2 e^{\frac{-2K\epsilon^2}{\log^2(1+L^2P)}}.
\end{align*}
The proof is concluded by exchanging limit and expectation based, e.g., on the dominated convergence theorem (probabilities are trivially bounded). 

\subsection{Proof of Proposition~\ref{prop:R_worstcase}}
\label{proof:R_worstcase}
For all $k\in \{0,\ldots,K-1\}$ and $l'\in \{1,\ldots,L\}$, we define $X \eqdef \sum_{l\neq l'} H_l[k]$ and write $\E\left[\log\left(1+\left|\textstyle \sum_{l=1}^LH_l[k]\right|^2P\right)\right]$
\begin{align*}
& = \E\left[\log\left(1+\left|X+H_{l'}[k]\right|^2P\right)\right] \\
&= \E\left[\E\left[\log\left(1+\left|X+H_{l'}[k]\right|^2P\right)\middle|X,\beta_{l'}\right]\right].
\end{align*}
Using Lemma~\ref{lem:closedform} and Lemma~\ref{lem:monotone} in Appendix~\ref{app:monotonicity}, we notice that the inner expectation is an increasing function of $\beta_{l'}|G_l[k]|^2$. Hence, it is minimized by letting $\delta_{l'}=0.5$ as in the proof of Proposition~\ref{prop:C_worstcase} in Appendix~\ref{proof:C_worstcase}. The second part follows similar steps as the second part of the proof of Proposition~\ref{prop:C_worstcase}, by replacing $\alpha$ with $|h|^2$. The limit and expectation can be exchanged based, e.g., on the dominated convergence theorem.

\bibliographystyle{IEEEbib}
\bibliography{IEEEabrv,refs}

\end{document}